\documentclass[reqno,a4paper,11pt]{article}
\pdfoutput=1

\usepackage{xcolor}
\usepackage{graphicx}
\usepackage[textwidth = 430 pt, textheight = 630 pt]{geometry}

\definecolor{MyDarkBlue}{rgb}{0.15,0.25,0.45}
\usepackage{epsfig,rotating}
\usepackage{amsmath,amssymb,amsbsy}
\usepackage{amsfonts}
\usepackage{mathrsfs}
\usepackage{bbm}
\usepackage[normalem]{ulem}

\usepackage{latexsym}
\usepackage{amsthm}
\usepackage[all,knot]{xy}
\xyoption{arc}

\usepackage[utf8x]{inputenc}

\usepackage{hyperref}
\hypersetup{
    hypertexnames=false,
    colorlinks=true,
    citecolor=MyDarkBlue,
    linkcolor=MyDarkBlue,
    urlcolor=MyDarkBlue,
    pdfauthor={Christian S\"amann and Emmanouil Sfinarolakis},
    pdftitle={Symmetry Factors of Feynman Diagrams and the Homological Perturbation Lemma},
    pdfsubject={hep-th math-ph},
    breaklinks=true
}

\usepackage{tikz}
\usetikzlibrary{matrix,cd,arrows,decorations.pathmorphing,external}
\usepackage{mathtools}
\usepackage[all,knot]{xy}
\xyoption{arc}

\tikzset{
    int/.style={
        line width=1.5pt,
        line cap=round,
        dash pattern=on 0pt off 3pt
    },
    ambi/.style={
        line width=1.2pt,
        line cap=round,
        dash pattern=on 2pt off 3pt
    },
    anti/.style={decorate, draw=black,
        decoration={zigzag,segment length = 1mm, amplitude = 0.7mm}}
}

\let\fn\footnote
\renewcommand{\footnote}[1]{\linespread{1.1}\fn{#1}\linespread{1.29}}

\makeatletter\renewcommand{\section}{\@startsection
    {section}{1}{\z@}{-3.5ex plus -1ex minus
        -.2ex}{2.3ex plus .2ex}{\bf\mathversion{bold} }}
\makeatletter\renewcommand{\subsection}{\@startsection{subsection}{2}{\z@}{-3.25ex
        plus -1ex minus
        -.2ex}{1.5ex plus .2ex}{\bf\mathversion{bold} }}
\makeatletter\renewcommand{\subsubsection}{\@startsection{subsubsection}{3}{-2.45ex}{-3.25ex
        plus -1ex minus -.2ex}{1.5ex plus .2ex}{\it }}
\renewcommand{\thesection}{\arabic{section}}
\renewcommand{\thesubsection}{\arabic{section}.\arabic{subsection}}
\renewcommand{\@seccntformat}[1]{\@nameuse{the#1}.~~}

\renewcommand{\theequation}{\thesection.\arabic{equation}}
\makeatletter \@addtoreset{equation}{section}
\def\Ddots{\mathinner{\mkern1mu\raise\p@
        \vbox{\kern7\p@\hbox{.}}\mkern2mu
        \raise4\p@\hbox{.}\mkern2mu\raise7\p@\hbox{.}\mkern1mu}}
\usepackage[toc,page]{appendix}

\newtheorem{thm}{Theorem}[section]
\renewcommand{\thethm}{\thesection.\arabic{thm}}
\newtheorem{lemma}[thm]{Lemma}
\newtheorem{algorithm}[thm]{Algorithm}

\newtheorem{theorem}[thm]{Theorem}

\newtheorem{corollary}[thm]{Corollary}

\newcommand{\wave}{\mathop\square}

\renewcommand{\appendices}{
    \section*{Appendix}\label{appendices}\setcounter{subsection}{0}
    \addcontentsline{toc}{section}{Appendix}
    \setcounter{equation}{0}
    \makeatletter
    \renewcommand{\theequation}{\Alph{subsection}.\arabic{equation}}
    \renewcommand{\thesubsection}{\Alph{subsection}}
    \renewcommand{\thethm}{\Alph{subsection}.\arabic{thm}}
    \@addtoreset{equation}{subsection}
    \@addtoreset{thm}{subsection}
    \makeatother
}

\makeatother

\newcommand{\makecommand}[3]{%
    \foreach \i in #3 {%
        \expandafter\xdef\csname #1\i\endcsname{\noexpand#2{\unexpanded\expandafter{\i}}}%
    }%
}
\newcommand{\latinalphabet}{A,a,B,b,C,c,d,D,E,e,F,f,G,g,H,h,I,i,J,j,K,k,L,l,M,m,N,n,O,o,P,p,Q,q,R,r,S,s,T,t,U,u,V,v,W,w,X,x,Y,y,Z,z}
\makecommand{I}{\mathbbm}{\latinalphabet}
\makecommand{bf}{\mathbf}{\latinalphabet}
\makecommand{bm}{\bm}{\latinalphabet}
\makecommand{ca}{\mathcal}{\latinalphabet}
\makecommand{fr}{\mathfrak}{\latinalphabet}
\makecommand{rm}{\mathrm}{\latinalphabet}
\makecommand{sf}{\mathsf}{\latinalphabet}
\makecommand{sc}{\mathscr}{\latinalphabet}

\makecommand{rm}{\mathrm}{{free,int}}
\makecommand{sf}{\mathsf}{{id}}

\newcommand{\Iint}{\pmb{\pitchfork}}

\def\slasha#1{\setbox0=\hbox{$#1$}#1\hskip-\wd0\hbox to\wd0{\hss\sl/\/\hss}}

\def\periodb#1{\setbox0=\hbox{$#1$}#1\hskip-\wd0\hbox to\wd0{-}}

\newcommand{\delder}[1]{\frac{\delta}{\delta #1}}   		

\newcommand{\comment}[1]{}     				
\def\tyng(#1){\hbox{\tiny$\yng(#1)$}}			
\def\tyoung(#1){\hbox{\tiny$\young(#1)$}}			

\newcommand{\eand}{~~~\mbox{and}~~~}
\newcommand{\ewith}{~~~\mbox{with}~~~}

\newcommand{\beq}{\begin{eqnarray}}
    \newcommand{\eeq}{\end{eqnarray}}

\begin{document}
    \begin{titlepage}
        \begin{flushright}
            EMPG--20--18
        \end{flushright}
        \vskip2.0cm
        \begin{center}
            {\LARGE \bf Symmetry Factors of Feynman Diagrams\\[0.3cm] and the Homological Perturbation Lemma}
            \vskip1.5cm
            {\Large Christian Saemann and Emmanouil Sfinarolakis}
            \setcounter{footnote}{0}
            \renewcommand{\thefootnote}{\arabic{thefootnote}}
            \vskip1cm
            {\em Maxwell Institute for Mathematical Sciences\\
                Department of Mathematics, Heriot--Watt University\\
                Colin Maclaurin Building, Riccarton, Edinburgh EH14 4AS,
                U.K.}\\[0.5cm]
            {Email: {\ttfamily c.saemann@hw.ac.uk~,~es73@hw.ac.uk}}
        \end{center}
        \vskip1.0cm
        \begin{center}
            {\bf Abstract}
        \end{center}
        \begin{quote}
            We discuss the symmetry factors of Feynman diagrams of scalar field theories with polynomial potential. After giving a concise general formula for them, we present an elementary and direct proof that when computing scattering amplitudes using the homological perturbation lemma, each contributing Feynman diagram is indeed included with the correct symmetry factor.
        \end{quote}
    \end{titlepage}
    
    \tableofcontents

    \section{Introduction and results}
    
    The Batalin--Vilkovisky (BV) formalism~\cite{Batalin:1981jr,Schwarz:1992nx} gives a formulation of any classical field theory as a differential $Q_{\rm BV}=\{S_{\rm BV},-\}$ on the graded commutative algebra generated by fields, antifields, ghosts, antifields of ghosts etc.~and all of their derivatives. Here, $Q_{\rm BV}$ is called the {\em BV differential}, $S_{\rm BV}$ is the {\em BV action} and the Poisson bracket is induced by the canonical symplectic form on BV field space. This differential graded commutative algebra is precisely the dual of a {\em strong homotopy Lie algebra} or {\em $L_\infty$-algebra}, for short, cf.~\cite{Jurco:2018sby} and references therein. 
    
    Recall that an $L_\infty$-algebra $\frL$ consists of a graded vector space $\frL=\oplus_{n\in \IZ}\frL_n$ together with a set of higher products $\mu_n:\frL^{\wedge n}\rightarrow \frL$ for all $n\in \IN^+$. The lowest product $\mu_1$ is a differential. For an $L_\infty$-algebra arising from a field theory, it encodes the free, linearized aspects of the theory, such as gauge symmetries and equations of motion. All higher products $\mu_{n\geq 2}$ describe interactions, covariantizations of linearized gauge symmetries, etc.
    
    This perspective is useful for a number of applications, in particular for discussing perturbative quantum field theories, see e.g.~\cite{Nutzi:2018vkl,Reiterer:2019dys,Macrelli:2019afx,Arvanitakis:2019ald,Jurco:2019yfd,Borsten:2020XXX}. Any $L_\infty$-algebra comes with a family of quasi-isomorphic (read: equivalent) $L_\infty$-algebras known as {\em minimal models}. For an $L_\infty$-algebra corresponding to a field theory, the graded vector space underlying the minimal models contains the asymptotically free fields while the higher products are precisely the tree level scattering amplitudes~\cite{Kajiura:0306332,Jurco:2018sby,Nutzi:2018vkl,Macrelli:2019afx}. Explicitly, the minimal model is computed via the homological perturbation lemma~\cite{Gugenheim1989:aa,gugenheim1991perturbation,Crainic:0403266}, see also~\cite{Gwilliam:2012jg,JohnsonFreyd:2012ww,Doubek:2017naz}, and this lemma produces the tree level perturbation expansion usually employed in quantum field theory.
    
    As pointed out in~\cite{Doubek:2017naz,Pulmann:2016aa}, the BV formalism gives a clear indication of how to extend the homological perturbation lemma to capture also the full quantum case. This leads to a construction of minimal models of {\em quantum} or {\em loop homotopy algebras} as introduced in~\cite{Zwiebach:1992ie,Markl:1997bj}. While the focus in~\cite{Doubek:2017naz,Pulmann:2016aa} was mostly on computing the effective action via the homological perturbation lemma in the BV picture, the dual, coalgebra picture was explored in~\cite{Jurco:2019yfd}, where very general recursion relations for quantum scattering amplitudes were derived. Both perspectives generate the familiar expansion in Feynman diagrams, merely in two opposite and dual directions.
    
    An evident question in this context is if and how precisely the homological perturbation lemma reproduces the exact symmetry factors of the individual Feynman diagrams. It is the main goal of this paper to provide an elementary and direct answer to this question.
    
    We start with a review of scattering amplitudes, homotopy algebras and the homological perturbation lemma. We explain in detail how the homology algebraic approach leads to what we call {\em HPL-diagrams}, which are closely related to Feynman diagrams. To be explicit and expose the crucial aspects of the approach, we limit ourselves to real scalar field theories with polynomial potentials. This also allows us to use $A_\infty$-algebras over $L_\infty$-algebras, which makes our discussion even simpler and more direct. 
    
    Next, we give a concise formula for general symmetry factors of Feynman diagrams for general scalar field theories, using the usual generating functional approach and extending the result of~\cite{Palmer:2001vq}. We then begin the comparison with HPL-diagrams with a few examples, proving a number of lemmata about their shapes and symmetry factors. Altogether, these lemmata show that in the computation of quantum scattering amplitudes via the homological perturbation lemma, all Feynman diagrams are indeed taken into account with the correct symmetry factors.
    
    While our result is certainly not surprising, we hope that our direct and explicit discussion, which is complemented by many examples, is helpful for readers trying to understand the relation between the homological algebraic approach to scattering amplitudes and the traditional Feynman diagrammatic one.

    \section{Scattering amplitudes from the homological perturbation lemma}
    
    In this section, we briefly review how the homological perturbation lemma produces the scattering amplitudes of a quantum field theory, cf.~\cite{Doubek:2017naz,Jurco:2018sby,Nutzi:2018vkl,Macrelli:2019afx,Jurco:2019yfd} as well as~\cite{Jurco:2020yyu} for a pedagogical review. We highlight the effect of the antisymmetry in the contracting homotopy, which leads to the cancellation of many diagrams.
    
    \subsection{Scattering amplitudes and homotopy algebras}
    
    For simplicity, we focus on a general scalar field theory on $d$-dimensional Minkowski $\IR^{1,d-1}$ space with action
    \begin{equation}\label{eq:action}
        S=\int \rmd^d x~\big(\caL_0+\caL_{\rm int}\big)=\int \rmd^d x~\Big( \tfrac12 \phi(x)(\wave-m^2) \phi(x)+\sum_{k\geq 3} \tfrac{1}{k!}c_k \phi^k(x)\Big)~.
    \end{equation}
    As explained e.g.~in~\cite{Jurco:2018sby,Macrelli:2019afx}, this scalar field theory can be encoded in a cyclic $L_\infty$-algebra $\frL^S$, which is simultaneously a cyclic $A_\infty$-algebra $\frA^S$.
    
    Recall that an $L_\infty$-algebra $\frL$ comes with higher products $\mu_n:\frL^{\wedge n}\rightarrow \frL$ of degree~$|\mu_n|=2-n$, where $\frL$ is a graded vector space $\frL=\oplus_{n\in\IZ}\frL_n$. Similarly, an $A_\infty$-algebra $\frA=\oplus_{n\in\IZ} \frA_n$ comes with higher products $\frm_n:\frA^{\otimes n}\rightarrow \frA$ of degree~$|\sfm_n|=2-n$. The two types of higher products then satisfy a higher or homotopy version of the Jacobi identity and the associativity relation, respectively. The precise details of these relations are irrelevant for our discussion.
    
    In a generalization of the relation between matrix algebras and matrix Lie algebras, antisymmetrizing the higher products $\sfm_n$ of an $A_\infty$-algebra $\frA$ leads to higher products for an $L_\infty$-algebra on the same graded vector space:
    \begin{equation}
        \mu_n(a_1,\ldots,a_n)=\sum_{\sigma\in S_n} \chi(\sigma;a_1,\ldots,a_n) \sfm_n(a_{\sigma(1)},\ldots,a_{\sigma(n)})~,~~~a_i\in \frA~,
    \end{equation}
    where the sum is taken over all permutations of the $a_i$ and the Koszul sign $\chi(\sigma;\ell_1,\ldots,\ell_n)$ is the sign arising from permuting the graded arguments $a_{\sigma(1)},\ldots,a_{\sigma(n)}$ into canonical order.
    
    In the homotopy algebras $\frL^S$ and $\frA^S$ encoding the field theory~\eqref{eq:action}, the subspaces $\frL^S_1=\frA^S_1$ and $\frL^S_2=\frA^S_2$ are the field and antifield spaces, respectively. The non-trivial differentials $\mu_1=\sfm_1$ are given by the kinematical operator,
    \begin{equation}
        \mu_1(\phi)=\sfm_1(\phi)=(\wave-m^2)\phi\in \frL^S_2~,~~~\phi\in \frL^S_1~.
    \end{equation}
    Furthermore, the non-trivial higher products $\mu_2$, $\mu_3$, $\ldots$ of $\frL^S$ encode the various interactions:
    \begin{equation}
        \mu_n(\phi_1,\ldots,\phi_n)=n!\,\sfm_n(\phi_1,\ldots,\phi_n)=-c_{n+1}\phi_1\ldots\phi_n\in \frL^S_2~,~~~\phi_i\in \frL^S_1~.
    \end{equation}
    
    Both $\frL^S$ and $\frA^S$ come with non-degenerate bilinear pairings $\langle-,-\rangle$ of degree~$-3$, which pair fields, i.e.~elements of $\frL^S_1=\frA^S_1$, with antifields, i.e.~elements of $\frL^S_2=\frA^S_2$, in the canonical way, producing a number. These pairings are compatible with the higher products in that they are cyclic:
    \begin{equation}\label{eq:rel_mu_m}
        \langle \phi_{n+1},\mu_n(\phi_1,\ldots,\phi_n)\rangle=\langle \phi_j,\mu_n(\phi_{j+1},\ldots,\phi_{n+1},\phi_1,\ldots,\phi_{j-1})\rangle~.
    \end{equation}
    
    In general, each homotopy algebra comes with quasi-isomorphic (i.e.~equivalent) {\em minimal models} with vanishing differential. These minimal models are all isomorphic to each other in an ordinary sense, and we will therefore often speak of {\em the} minimal model. The equivalence between the minimal model $\frL^{S\circ}$ of $\frL^S$ and $\frL^S$ itself translates into a classical equivalence of the corresponding field theories and the only theory equivalent to~\eqref{eq:action} without kinematical term is one in which the interaction vertices are the tree level scattering amplitudes. Computing the tree level scattering amplitudes is therefore tantamount to computing the minimal model of the $L_\infty$-algebra $\frL^S$.
    
    Explicitly, the formula for tree level amplitudes is 
    \begin{equation}\label{eq:scattering_amplitude_tree_level}
        \caA(\phi_1,\ldots\phi_{n+1})=\langle \phi_{n+1},\mu^\circ_n(\phi_1,\ldots,\phi_n)\rangle^\circ=\sum_{\sigma\in S_n} \langle \phi_{n+1},\sfm^\circ_n(\ell_{\sigma(1)},\ldots,\ell_{\sigma(n)})\rangle^\circ~,
    \end{equation}
    where a $\circ$ indicates the higher products and the cyclic structure in $\frL^{S\circ}$ and $\frA^{S\circ}$. Note that cyclicity of the bilinear pairing corresponds to the usual cyclic permutation symmetry of scattering amplitudes.
    
    \subsection{The homological perturbation lemma at tree level}
    
    The homological perturbation lemma reproduces precisely the usual perturbative expansion of quantum field theory at tree level and, with a deformed perturbation, also at loop level. In the following, we outline this construction, starting with tree level. 
    
    The field space $\frF=\frL^S_1$ decomposes into a direct sum of asymptotic, on-shell fields $\frF_\rmfree$ given by the kernel of the kinematic operator $\mu_1$ as well as the interacting fields $\frF_\rmint$ which propagate in the interior of a Feynman diagram, cf.~\cite{Macrelli:2019afx}. The chain complex underlying $\frL^S$ (and $\frA^S$) is then
    \begin{equation}
        \dots \xrightarrow{~~~} * \xrightarrow{~~~} \underbrace{\frF_\rmfree \oplus \frF_\rmint}_{\frL^S_1/\rm fields} \xrightarrow{~\mu_1~} \underbrace{\frF_\rmfree \oplus \frF_\rmint}_{\frL^S_2/\rm antifields} \xrightarrow{~~~} * \xrightarrow{~~~} \dots~.
    \end{equation}
    We note that there is an inverse to $\mu_1$ on $\frF_\rmint\subset \frL^S_2$, the {\em propagator} $\sfh$, which we trivially continue to a map of degree~$-1$ on all of $\frL^S_2$ and $\frL^S$ by setting it to zero everywhere else. In particular, $\frF_\rmint$ is the image of $\sfh$ for scalar field theory.
    
    The minimal models $\frL^{S\circ}$ of $\frL^S$ (and $\frA^{S\circ}$ of $\frA^S$) have underlying chain complex
    \begin{equation}
        \ldots \xrightarrow{~~~} * \xrightarrow{~~~} \underbrace{\frF_\rmfree}_{\frL^{S\circ}_1/\rm fields} \xrightarrow{~0~} \underbrace{\frF_\rmfree}_{\frL^S_2/\rm antifields} \xrightarrow{~~~} * \xrightarrow{~~~} \ldots~,
    \end{equation}
    and we have the trivial projection and embedding
    \begin{subequations}
        \begin{equation}
            \sfp:\frL^S\rightarrow \frL^{S\circ}
            \eand 
            \sfe:\frL^{S\circ}\hookrightarrow\frL^{S}
            \ewith \sfp\circ \sfe=\sfid~.
        \end{equation}
        We can use the propagator to make the quasi-isomorphism between both chain complexes explicit. That is, $\sfh$ is a contracting homotopy between the identity on $\frL^S$ and the concatenation $\sfe\circ \sfp$:
        \begin{equation}\label{eq:contracting_homotopy}
            \sfid-\sfe\circ\sfp=\sfm_1\circ\sfh+\sfh\circ \sfm_1~.
        \end{equation}
        Furthermore, $\sfh$ can always be chosen such that 
        \begin{equation}\label{eq:contractingBasic}
            \sfp\circ\sfh = \sfh\circ\sfe = \sfh\circ\sfh = 0~,~~~\sfp\circ \sfm_1 = \sfm_1\circ\sfe = 0~.
        \end{equation}
    \end{subequations}
    The homological perturbation lemma (HPL)~\cite{Gugenheim1989:aa,gugenheim1991perturbation,Crainic:0403266} can now be used to lift this quasi-isomorphism of chain complexes $(\frL^S,\mu_1)\cong (\frL^{S\circ},0)$ to a quasi-isomorphism of $L_\infty$-algebras $(\frL^S,\mu_\bullet)\cong (\frL^{S\circ},\mu^\circ_{\bullet})$ by regarding $\mu_2$, $\mu_3$, $\ldots$ as perturbations. Since these describe precisely the interaction terms in~\eqref{eq:action}, and thus the perturbations of the free theory, it is not surprising that we recover the usual perturbative expansion.
    
    The treatment of perturbations is technically somewhat simpler\footnote{Note that this is a pedagogical choice. Conceptually, the discussion for $L_\infty$-algebra is also straightforward: one merely has to insert symmetrizers in all compositions of maps. These, however, would awkwardly distort and lengthen all formulas.} for $A_\infty$-algebras, and we restrict ourselves to these in the following.  Moreover, for the link to Feynman diagrams, it is helpful to discuss the HPL using the coalgebra description of $A_\infty$-algebras. Here, we consider the tensor algebra of the grade-shifted vector space $\sfV=\oplus_{n\in\IZ}\sfV_n$ with $\sfV_n\coloneqq\frL^S_{n+1}$,
    \begin{equation}
        \otimes^\bullet \sfV=\bigoplus_{k=0}^\infty \sfV^{\otimes k}=\IR\oplus \sfV\oplus (\sfV\otimes \sfV)\oplus \ldots~.
    \end{equation}
    Note that in the case of scalar field theory, only $\sfV_0$ and $\sfV_1$ are non-trivial, describing fields and antifields, respectively. Taking into account the degree shift, the higher products $\sfm_n$ are now all of degree~$1$ and can be extended to coderivations,
    \begin{equation}
        \sfM_k(\phi_1,\ldots,\phi_n)=\sum_{j=0}^{n-k}\phi_{1}\otimes \ldots \otimes \phi_j\otimes \sfm_k(\phi_{j+1},\ldots,\phi_{j+k})\otimes \phi_{j+i+1}\otimes \ldots \otimes \phi_n~,
    \end{equation}
    which combine into the {\em codifferential}
    \begin{equation}
        \sfD=\sfD_0+\sfD_\rmint~,~~~\sfD_0=\sfM_1\eand \sfD_\rmint=\sfM_2+\sfM_3+\ldots~.
    \end{equation}
    This codifferential satisfies $\sfD^2=0$, which is equivalent to the homotopy associativity conditions.
    
    The maps $\sfp$, $\sfe$ and $\sfh$ are readily lifted to the coalgebra picture as follows\footnote{This is known as the ``tensor trick.''}:
    \begin{equation}\label{eq:lift_to_tensor_algebra}
        \sfP_0|_{\bigotimes^k \sfV} \coloneqq \sfp^{\otimes k}~,~~~
        \sfE_0|_{\bigotimes^k \sfV^\circ} \coloneqq \sfe^{\otimes k}~,~~~~\sfH_0|_{\bigotimes^k \sfV} \coloneqq \sum_{i+j=k-1}\sfid^{\otimes i}\otimes\sfh\otimes(\sfe\circ\sfp)^{\otimes j}~,
    \end{equation}
    and we have the relations
    \begin{equation}\label{eq:contractingBasic_H0}
        \begin{gathered}
            \sfid-\sfE_0\circ\sfP_0 = \sfD_0\circ\sfH_0+\sfH_0\circ \sfD_0~,\\
            \sfP_0\circ\sfE_0 = \sfid~,~~~
            \sfP_0\circ\sfH_0 = \sfH_0\circ\sfE_0 = \sfH_0\circ\sfH_0 = 0~,~~~
            \sfP_0\circ \sfD_0 = \sfD_0\circ\sfE_0 = 0~.
        \end{gathered}
    \end{equation}
    
    The homological perturbation lemma then states that we also have a contracting homotopy $\sfH$ after perturbing $\sfD_0$ to $\sfD=\sfD_0+\sfD_\rmint$ with the following maps:
    \begin{equation}\label{eq:hpl_relations}
        \begin{aligned}
            \sfP &= \sfP_0\circ(\sfid+\sfD_\rmint\circ\sfH_0)^{-1},~~~&
            \sfH &= \sfH_0\circ(\sfid+\sfD_\rmint\circ\sfH_0)^{-1}~,\\
            \sfE &= (\sfid+\sfH_0\circ\sfD_\rmint)^{-1}\circ\sfE_0~,~~~&
            \sfD^\circ &= \sfP\circ\sfD_\rmint\circ\sfE_0~,
        \end{aligned}
    \end{equation}
    where
    \begin{equation}
        \sfD^\circ=\sfD^\circ_2+\sfD^\circ_3+\ldots~,
    \end{equation}
    is the full codifferential on the minimal model. The inverse of the map $\sfid+\sfD_\rmint\circ\sfH_0$ exist because $\sfD_\rmint$ is a small perturbation, and this inverse is to be regarded as the evident geometric series.    
    
    Recall that $\sfD^\circ$ encodes the higher products $\sfm_n^\circ$ of the minimal model and thus the tree level scattering amplitudes~\eqref{eq:scattering_amplitude_tree_level}. Because of~\eqref{eq:hpl_relations}, we can compute $\sfD^\circ$ from the recursion 
    \begin{equation}\label{eq:tree_recursion}
        \sfD^\circ=\sfP_0\circ\sfD_\rmint\circ\sfE~,~~~\sfE=\sfE_0-\Iint\!\circ\sfE\ewith \Iint\coloneqq\sfH_0\circ\sfD_\rmint~.
    \end{equation}
    
    \subsection{Example: the four-point amplitude}
    
    To illustrate the above constructions, let us discuss an example. For computing the four-point amplitude $\caA(\phi_1,\ldots,\phi_4)$, we need to determine $\sfm_3^\circ$, cf.~the general formula~\eqref{eq:scattering_amplitude_tree_level}. This higher product can be extracted from $\sfD^\circ$, if we know the latter's action on three input fields. The relevant terms in the recursion for $\sfD^\circ$ are
    \begin{equation}\label{eq:recursion_4pt}
        \sfP_0\circ \sfD_\rmint \circ (\sfE_0-\sfH_0\circ \sfD_\rmint\circ \sfE_0)=\sfP_0\circ \sfD_\rmint \circ (\sfE_0-\Iint\circ \sfE_0)~,
    \end{equation}
    which is a truncation of~\eqref{eq:tree_recursion} after one iteration. We represent an element of $\sfV^\circ_0\otimes \sfV^\circ_0\otimes \sfV^\circ_0$ diagrammatically in what we call an {\em HPL diagram} by three parallel lines (HPL diagrams are always read from bottom to top), 
    \begin{equation}
        \begin{tikzpicture}[baseline={([yshift=-.5ex]current bounding box.center)}]
            \matrix (m) [matrix of nodes, ampersand replacement=\&, column sep = 0.2cm, row sep = 0.2cm]{
                {} \& {} \& {} \\
                {} \& {} \& {} \\
                $\phi_1$ \& $\phi_2$ \& $\phi_3$ \\
            };
            \draw (m-1-1) -- (m-3-1) ;
            \draw (m-1-2) -- (m-3-2) ;
            \draw (m-1-3) -- (m-3-3) ;
        \end{tikzpicture}
    \end{equation}
    For interactions to take place, we have to use $\sfE_0$ to embed $\sfV^\circ_0\otimes \sfV^\circ_0\otimes \sfV^\circ_0$ in $\sfV\otimes \sfV\otimes \sfV$, but we usually do not label this embedding separately in diagrams:
    \begin{equation}
        \begin{tikzpicture}[baseline={([yshift=-.5ex]current bounding box.center)}]
            \matrix (m) [matrix of nodes, ampersand replacement=\&, column sep = 0.2cm, row sep = 0.2cm]{
                {} \& {} \& {} \\
                $\sfe$ \& $\sfe$ \& $\sfe$ \\
                $\phi_1$ \& $\phi_2$ \& $\phi_3$ \\
            };
            \draw (m-1-1) -- (m-2-1) ;
            \draw (m-1-2) -- (m-2-2) ;
            \draw (m-1-3) -- (m-2-3) ;
            \draw (m-2-1) -- (m-3-1) ;
            \draw (m-2-2) -- (m-3-2) ;
            \draw (m-2-3) -- (m-3-3) ;
        \end{tikzpicture}
        ~=~
        \begin{tikzpicture}[baseline={([yshift=-.5ex]current bounding box.center)}]
            \matrix (m) [matrix of nodes, ampersand replacement=\&, column sep = 0.2cm, row sep = 0.2cm]{
                {} \& {} \& {} \\
                {} \& {} \& {} \\
                $\phi_1$ \& $\phi_2$ \& $\phi_3$ \\
            };
            \draw (m-1-1) -- (m-3-1) ;
            \draw (m-1-2) -- (m-3-2) ;
            \draw (m-1-3) -- (m-3-3) ;
        \end{tikzpicture}
    \end{equation}
    Similarly, we do not depict the projection $\sfp$ from $\sfV_1$ down to $\sfV^\circ_1$ in any diagram.
    
    On $\sfV\otimes \sfV\otimes \sfV$, the action of $\sfD_\rmint$ agrees with $\sfM_2+\sfM_3$ and on this subspace of the tensor algebra, we can depict these operators in HPL diagrams as follows:
    \begin{equation}
        \sfM_2=
        ~\begin{tikzpicture}[baseline={([yshift=-.5ex]current bounding box.center)}]
            \matrix (m) [matrix of nodes, ampersand replacement=\&, column sep = 0.15cm, row sep = 0.2cm]{
                {} \& {} \& {} \& {} \\
                {} \& {} \& {} \& {} \\
                {} \& {} \& {} \& {} \\
            };
            \draw (m-2-2.center) -- (m-3-1) ;
            \draw[anti] (m-1-2) -- (m-2-2.center) ;
            \draw (m-2-2.center) -- (m-3-3) ;
            \draw (m-1-4) -- (m-3-4) ;
        \end{tikzpicture}
        ~+~
        \begin{tikzpicture}[baseline={([yshift=-.5ex]current bounding box.center)}]
            \matrix (m) [matrix of nodes, ampersand replacement=\&, column sep = 0.15cm, row sep = 0.2cm]{
                {} \& {} \& {} \& {} \\
                {} \& {} \& {} \& {} \\
                {} \& {} \& {} \& {} \\
            };
            \draw (m-2-3.center) -- (m-3-2) ;
            \draw[anti] (m-1-3) -- (m-2-3.center) ;
            \draw (m-2-3.center) -- (m-3-4) ;
            \draw (m-1-1) -- (m-3-1) ;
        \end{tikzpicture}
        ~~~,~~~
        \sfM_3=
        \begin{tikzpicture}[baseline={([yshift=-.5ex]current bounding box.center)}]
            \matrix (m) [matrix of nodes, ampersand replacement=\&, column sep = 0.2cm, row sep = 0.2cm]{
                {} \& {} \& {} \\
                {} \& {} \& {} \\
                {} \& {} \& {} \\
            };
            \draw (m-2-2.center) -- (m-3-1) ;
            \draw[anti] (m-1-2) -- (m-2-2.center) ;
            \draw (m-2-2.center) -- (m-3-3) ;
            \draw (m-2-2.center) -- (m-3-2) ;
        \end{tikzpicture}
        ~.
    \end{equation}
    Here, a zigzag line denotes an antifield, i.e.~an element in $\sfV_1$. While the antifield resulting from $\sfM_3$ will be paired off with a field in the final amplitude, the antifields produced by $\sfM_2$ are turned back into fields by a propagator $\sfh$ produced by
    \begin{equation}
        \sfH_0|_{\sfV\otimes \sfV}=\sfh\otimes (\sfe\circ\sfp)+\sfid\otimes \sfh~.
    \end{equation}
    In the composition $\sfH_0\circ \sfM_2$, we produce two different types of fields: the straight line stands for an element in $\sfe(\sfV_0)\subset \sfV$, and on this subspace $\sfid=\sfe\circ \sfp$ because of $\sfe\circ \sfp\circ \sfe=\sfe$ which is a consequence of~\eqref{eq:contractingBasic}. The straight line is thus simply an on-shell field, i.e.~an element of $\frF_\rmfree$. Because of $\sfh\circ \sfe=0$, however, the output of a propagator is always an interacting field, i.e.~an element of $\frF_\rmint$. It will be important to distinguish these, and in HPL diagrams, we depict fields in $\frF_\rmint$ by a dotted line:
    \begin{equation}
        \sfH_0\circ \sfM_2=
        ~\begin{tikzpicture}[baseline={([yshift=-.5ex]current bounding box.center)}]
            \matrix (m) [matrix of nodes, ampersand replacement=\&, column sep = 0.15cm, row sep = 0.2cm]{
                {} \& {} \& {} \& {} \\
                {} \& $\sfh$ \& {} \& {} \\
                {} \& {} \& {} \& {} \\
                {} \& {} \& {} \& {} \\
            };
            \draw[int] (m-1-2) -- (m-2-2) ;
            \draw (m-3-2.center) -- (m-4-1) ;
            \draw[anti] (m-2-2) -- (m-3-2.center) ;
            \draw (m-3-2.center) -- (m-4-3) ;
            \draw (m-1-4) -- (m-4-4) ;
        \end{tikzpicture}
        ~+~
        \begin{tikzpicture}[baseline={([yshift=-.5ex]current bounding box.center)}]
            \matrix (m) [matrix of nodes, ampersand replacement=\&, column sep = 0.15cm, row sep = 0.2cm]{
                {} \& {} \& {} \& {} \\
                {} \& {} \& $\sfh$ \& {} \\
                {} \& {} \& {} \& {} \\
                {} \& {} \& {} \& {} \\
            };
            \draw[int] (m-1-3) -- (m-2-3) ;
            \draw (m-3-3.center) -- (m-4-2) ;
            \draw[anti] (m-2-3) -- (m-3-3.center) ;
            \draw (m-3-3.center) -- (m-4-4) ;
            \draw (m-1-1) -- (m-4-1) ;
        \end{tikzpicture}
        ~=
        ~\begin{tikzpicture}[baseline={([yshift=-.5ex]current bounding box.center)}]
            \matrix (m) [matrix of nodes, ampersand replacement=\&, column sep = 0.15cm, row sep = 0.2cm]{
                {} \& {} \& {} \& {} \\
                {} \& {} \& {} \& {} \\
                {} \& {} \& {} \& {} \\
            };
            \draw (m-2-2.center) -- (m-3-1) ;
            \draw[int] (m-1-2) -- (m-2-2.center) ;
            \draw (m-2-2.center) -- (m-3-3) ;
            \draw (m-1-4) -- (m-3-4) ;
        \end{tikzpicture}
        ~+~
        \begin{tikzpicture}[baseline={([yshift=-.5ex]current bounding box.center)}]
            \matrix (m) [matrix of nodes, ampersand replacement=\&, column sep = 0.15cm, row sep = 0.2cm]{
                {} \& {} \& {} \& {} \\
                {} \& {} \& {} \& {} \\
                {} \& {} \& {} \& {} \\
            };
            \draw (m-2-3.center) -- (m-3-2) ;
            \draw[int] (m-1-3) -- (m-2-3.center) ;
            \draw (m-2-3.center) -- (m-3-4) ;
            \draw (m-1-1) -- (m-3-1) ;
        \end{tikzpicture}
    \end{equation}
    Together with the facts that $\sfD_\rmint$ vanishes on $\sfV$ (single fields cannot interact) and that antifields are paired off with fields in $\sfV^\circ$ by the cyclic pairing, the formula~\eqref{eq:recursion_4pt} produces HPL diagrams corresponding to three of the usual four Feynman diagrams for four-point amplitudes:
    \begin{equation}
        \begin{tikzpicture}[baseline={([yshift=-.5ex]current bounding box.center)}]
            \matrix (m) [matrix of nodes, ampersand replacement=\&, column sep = 0.15cm, row sep = 0.2cm]{
                {} \& {} \& $\phi_4$  \& {} \\
                {} \& {} \& {} \& {} \\
                {} \& {} \& {} \& {} \\
                {} \& {} \& {} \& {} \\
                $\phi_1$ \& {} \& $\phi_2$  \& $\phi_3$ \\
            };
            \draw[int] (m-2-3.center) -- (m-3-2.center) ;
            \draw[anti] (m-1-3) -- (m-2-3.center) ;
            \draw (m-2-3.center) -- (m-3-4.center) ;
            \draw (m-4-2.center) -- (m-5-1) ;
            \draw[int] (m-3-2.center) -- (m-4-2.center) ;
            \draw (m-4-2.center) -- (m-5-3) ;
            \draw (m-3-4.center) -- (m-5-4) ;
        \end{tikzpicture}
        ~+~
        \begin{tikzpicture}[baseline={([yshift=-.5ex]current bounding box.center)}]
            \matrix (m) [matrix of nodes, ampersand replacement=\&, column sep = 0.15cm, row sep = 0.2cm]{
                {} \& $\phi_4$  \& {} \& {} \\
                {} \& {} \& {} \& {} \\
                {} \& {} \& {} \& {} \\
                {} \& {} \& {} \& {} \\
                $\phi_1$  \& $\phi_2$  \& {} \& $\phi_3$  \\
            };
            \draw (m-2-2.center) -- (m-3-1.center) ;
            \draw[anti] (m-1-2) -- (m-2-2.center) ;
            \draw[int] (m-2-2.center) -- (m-3-3.center) ;
            \draw (m-4-3.center) -- (m-5-2) ;
            \draw[int] (m-3-3.center) -- (m-4-3.center) ;
            \draw (m-4-3.center) -- (m-5-4) ;
            \draw (m-3-1.center) -- (m-5-1) ;
        \end{tikzpicture}
        ~+~
        \begin{tikzpicture}[baseline={([yshift=-.5ex]current bounding box.center)}]
            \matrix (m) [matrix of nodes, ampersand replacement=\&, column sep = 0.2cm, row sep = 0.2cm]{
                {} \& $\phi_4$ \& {} \\
                {} \& {} \& {} \\
                $\phi_1$ \& $\phi_2$ \& $\phi_3$ \\
            };
            \draw (m-2-2.center) -- (m-3-1) ;
            \draw[anti] (m-1-2) -- (m-2-2.center) ;
            \draw (m-2-2.center) -- (m-3-3) ;
            \draw (m-2-2.center) -- (m-3-2) ;
        \end{tikzpicture}
    \end{equation}
    The remaining diagram (as well as copies of the above diagrams to obtain the correct total factors) are produced by the permutations of the fields $\phi_{1,2,3}$ in formula~\eqref{eq:scattering_amplitude_tree_level}.
    
    An important point to notice is that the asymmetry in the definition of $\sfH_0$ in~\eqref{eq:lift_to_tensor_algebra} between left and right sides of the propagator $\sfh$ implies cancellations. Any interaction vertex to the left of a dotted line leads to an operator $\sfe\circ \sfp$ acting on an element of $\frF_\rmint={\rm im}(\sfh)$, which vanishes. That is, any summand in the recursion~\eqref{eq:tree_recursion} that is depicted by an HPL diagram containing rows of the form 
    \begin{equation}\label{eq:vanish_diag}
        \begin{tikzpicture}[baseline={([yshift=-.5ex]current bounding box.center)}]
            \matrix (m) [matrix of nodes, ampersand replacement=\&, column sep = 0.15cm, row sep = 0.2cm]{
                {} \& {} \& {} \& {} \& {} \& {} \& {}\\
                $\sfid$ \& {} \& $\sfid$ \& {} \& $\sfh$ \& {} \& $\sfe\circ\sfp$\\
                {} \& $\cdots$ \& {} \& {} \& {} \& {} \& {}\\
                {} \& {} \& {} \& {} \& {} \& {} \& {}\\
            };
            \draw[ambi] (m-3-5.center) -- (m-4-4) ;
            \draw[int] (m-1-5) -- (m-2-5) ;
            \draw[anti] (m-2-5) -- (m-3-5.center) ;
            \draw[ambi] (m-3-5.center) -- (m-4-6) ;
            \draw[ambi] (m-2-1) -- (m-4-1) ;
            \draw[ambi] (m-2-3) -- (m-4-3) ;
            \draw[ambi] (m-1-1) -- (m-2-1) ;
            \draw[ambi] (m-1-3) -- (m-2-3) ;
            \draw (m-1-7) -- (m-2-7) ;
            \draw[int] (m-2-7) -- (m-4-7) ;
        \end{tikzpicture}
        =
        \begin{tikzpicture}[baseline={([yshift=-.5ex]current bounding box.center)}]
            \matrix (m) [matrix of nodes, ampersand replacement=\&, column sep = 0.15cm, row sep = 0.2cm]{
                {} \& {} \& {} \& {} \& {} \& {} \& {}\\
                {} \& $\cdots$ \& {} \& {} \& {} \& {} \& {}\\
                {} \& {} \& {} \& {} \& {} \& {} \& {}\\
            };
            \draw[ambi] (m-2-5.center) -- (m-3-4) ;
            \draw[int] (m-1-5) -- (m-2-5.center) ;
            \draw[ambi] (m-2-5.center) -- (m-3-6) ;
            \draw[ambi] (m-1-1) -- (m-3-1) ;
            \draw[ambi] (m-1-3) -- (m-3-3) ;
            \draw[int] (m-1-7) -- (m-3-7) ;
        \end{tikzpicture}
        ~~\rightarrow~~0
    \end{equation}
    necessarily vanishes. Here, the dashed lines represent an arbitrary element in $\frF_\rmfree\oplus \frF_\rmint$.
    
    \subsection{Loop level scattering amplitudes}
    
    After the above discussion, the transition to loop or full quantum level is rather straightforward. Note that we will completely ignore any issues related to regularization; we are merely interested in combinatorial aspects and the relation of our construction to Feynman diagrams.
    
    Recall that in the BV formalism, switching from the classical to the quantum master equation essentially amounts to replacing the BV differential by the sum of the BV differential and the BV Laplacian. We are in the dual, coalgebra picture, and we thus introduce the operator $\rmi\hbar \Delta^*$ which inserts a full pair of fields and antifields:
    \begin{equation}\label{eq:def_Delta}
        \rmi\hbar \Delta^*(\phi_1\otimes \ldots \otimes \phi_n)\coloneqq
        \sum_{i=0}^n\sum_{j=i}^n\phi_1\otimes\ldots \otimes\phi_i\otimes \phi^A\otimes \phi_{i+1}\ldots \phi_j\otimes\phi_A\otimes \phi_{j+1}\otimes \ldots\otimes\phi_n~.
    \end{equation}
    Here, $A$ is a DeWitt indices which, for scalar field theory, can be taken to be a momentum mode label combined with a label distinguishing fields from antifields. Just as for $\sfD_\rmint$, a subsequent operator $\sfH_0$ then maps the antifield back to fields. Due to the asymmetry in $\sfH$, the contribution of $\rmi\hbar \Delta^*$ is only non-vanishing, if $\phi^A$ is a field and $\phi_A$ is an antifield in~\eqref{eq:def_Delta}. For simplicity, we define
    \begin{equation}
        \sfU\coloneqq\sfH_0\circ(\rmi\hbar \Delta^*)
    \end{equation}
    and in terms of HPL diagrams, we have e.g.~
    \begin{equation}
        \begin{aligned}
            \sfU\left(
            \begin{tikzpicture}[baseline={([yshift=-.5ex]current bounding box.center)}]
                \matrix (m) [matrix of nodes, ampersand replacement=\&, column sep = 0.2cm, row sep = 0.2cm]{
                    {} \\
                    {} \\
                    $\phi_1$ \\
                };
                \draw (m-1-1) -- (m-3-1) ;
            \end{tikzpicture}
            \right)~&=~
            \begin{tikzpicture}[baseline={([yshift=-.5ex]current bounding box.center)}]
                \matrix (m) [matrix of nodes, ampersand replacement=\&, column sep = 0.2cm, row sep = 0.2cm]{
                    {} \& {} \& {} \& {} \\
                    $\sfid$ \& {} \& $\sfh$ \& $\sfe\circ\sfp$ \\
                    {} \& {} \& {} \& {} \\
                    {} \& {} \& {} \& $\phi_1$ \\
                };
                \draw[int] (m-1-1) -- (m-2-1) ;
                \draw[int] (m-2-1) -- (m-3-1.center) ;
                \draw[int] (m-3-1.center) -- (m-3-2.center) ;
                \draw[anti] (m-3-2.center) -- (m-3-3.center) ;
                \draw[int] (m-1-3) -- (m-2-3) ;
                \draw[anti] (m-2-3) -- (m-3-3.center) ;
                \draw (m-1-4) -- (m-2-4) ;
                \draw (m-2-4) -- (m-4-4) ;
            \end{tikzpicture}
            ~+~
            \begin{tikzpicture}[baseline={([yshift=-.5ex]current bounding box.center)}]
                \matrix (m) [matrix of nodes, ampersand replacement=\&, column sep = 0.2cm, row sep = 0.2cm]{
                    {} \& {} \& {} \\
                    $\sfid$ \& $\sfid$ \& $\sfh$ \\
                    {} \& {} \& {} \\
                    {} \& $\phi_1$ \& {} \\
                };
                \draw[int] (m-1-1) -- (m-2-1) ;
                \draw[int] (m-2-1) -- (m-3-1.center) ;
                \draw[int] (m-3-1.center) -- (m-3-2.center) ;
                \draw[anti] (m-3-2.center) -- (m-3-3.center) ;
                \draw[int] (m-1-3) -- (m-2-3) ;
                \draw[anti] (m-2-3) -- (m-3-3.center) ;
                \draw (m-1-2) -- (m-2-2) ;
                \draw (m-2-2) -- (m-4-2) ;
            \end{tikzpicture}
            ~+~
            \begin{tikzpicture}[baseline={([yshift=-.5ex]current bounding box.center)}]
                \matrix (m) [matrix of nodes, ampersand replacement=\&, column sep = 0.2cm, row sep = 0.2cm]{
                    {} \& {} \& {} \& {} \\
                    $\sfid$ \& $\sfid$ \& {} \& $\sfh$ \\
                    {} \& {} \& {} \& {} \\
                    $\phi_1$ \& {} \& {} \& {} \\
                };
                \draw[int] (m-1-2) -- (m-2-2) ;
                \draw[int] (m-2-2) -- (m-3-2.center) ;
                \draw[int] (m-3-2.center) -- (m-3-3.center) ;
                \draw[anti] (m-3-3.center) -- (m-3-4.center) ;
                \draw[int] (m-1-4) -- (m-2-4) ;
                \draw[anti] (m-2-4) -- (m-3-4.center) ;
                \draw (m-1-1) -- (m-2-1) ;
                \draw (m-2-1) -- (m-4-1) ;
            \end{tikzpicture}\\
            &=~
            \begin{tikzpicture}[baseline={([yshift=-.5ex]current bounding box.center)}]
                \matrix (m) [matrix of nodes, ampersand replacement=\&, column sep = 0.2cm, row sep = 0.2cm]{
                    {} \& {} \& {} \& {} \\
                    {} \& {} \& {} \& {} \\
                    {} \& {} \& {} \& {} \\
                    {} \& {} \& {} \& $\phi_1$ \\
                };
                \draw[int] (m-3-1.center) -- (m-1-1) ;
                \draw[int] (m-3-1.center) -- (m-3-3.center) ;
                \draw[int] (m-1-3) -- (m-3-3.center) ;
                \draw (m-1-4) -- (m-4-4) ;
            \end{tikzpicture}
            ~+~
            \begin{tikzpicture}[baseline={([yshift=-.5ex]current bounding box.center)}]
                \matrix (m) [matrix of nodes, ampersand replacement=\&, column sep = 0.2cm, row sep = 0.2cm]{
                    {} \& {} \& {} \\
                    {} \& {} \& {} \\
                    {} \& {} \& {} \\
                    {} \& $\phi_1$ \& {} \\
                };
                \draw[int] (m-3-1.center) -- (m-1-1) ;
                \draw[int] (m-3-1.center) -- (m-3-3.center) ;
                \draw[int] (m-1-3) -- (m-3-3.center) ;
                \draw (m-1-2) -- (m-4-2) ;
            \end{tikzpicture}
            ~+~
            \begin{tikzpicture}[baseline={([yshift=-.5ex]current bounding box.center)}]
                \matrix (m) [matrix of nodes, ampersand replacement=\&, column sep = 0.2cm, row sep = 0.2cm]{
                    {} \& {} \& {} \& {} \\
                    {} \& {} \& {} \& {} \\
                    {} \& {} \& {} \& {} \\
                    $\phi_1$ \& {} \& {} \& {} \\
                };
                \draw[int] (m-3-2.center) -- (m-1-2) ;
                \draw[int] (m-3-2.center) -- (m-3-4.center) ;
                \draw[int] (m-3-4.center) -- (m-1-4) ;
                \draw (m-1-1) -- (m-4-1) ;
            \end{tikzpicture}
        \end{aligned}
    \end{equation}
    Note that the DeWitt index $A$ in~\eqref{eq:def_Delta} can be chosen such that it splits into a field/antifield label and two labels $\vec{p}_\rmfree$ and $p_\rmint$, where $\vec{p}_\rmfree\in\IR^{d-1}$ labels a free on-shell field with corresponding $(d-1)$-momentum while $p_\rmint\in\IR^{1,d-1}$ labels an interacting field with corresponding four-momentum. The presence of the propagator $\sfh$ in $\sfU$, however, annihilates the free on-shell fields and therefore the fields produced by $\sfU$ are always in $\frF_\rmint$.
    
    The scattering amplitudes are then extracted from formula~\eqref{eq:scattering_amplitude_tree_level}, but the higher products $\mu_n$ and $\sfm_n$ are now those combining into an operator\footnote{The perturbation by the second order differential operator $\Delta^*$ implies that $\sfD^\circ$ no longer defines an ordinary homotopy algebra, but a quantum or loop homotopy algebra. Also, the other maps $\sfP$ and $\sfE$ appearing in the homological perturbation lemma are no longer ordinary coalgebra morphisms.}
    $\sfD^\circ$ computed by the recursion
    \begin{equation}\label{eq:quantum_recursion}
        \sfD^\circ=\sfP_0\circ\sfD_\rmint\circ\sfE~,~~~\sfE=\sfE_0-\Iint\!\circ\sfE-\sfU\circ\sfE~.
    \end{equation}
    Each $\sfU$ operator produces a loop, and the number of loops is therefore counted by the powers of $\hbar$.
    
    As a reasonably simple example, let us consider how the HPL produces the two-loop contribution to the 2-point amplitude. (In the quantum case, the 2-point amplitude encoded by $\sfm_1^\circ$ only vanishes to zeroth order in $\hbar$.) Specifically, let us consider the diagram
    \begin{equation}\label{diag:ex_1}
        \begin{tikzpicture}[baseline={([yshift=-.5ex]current bounding box.center)}]
            \matrix (m) [matrix of nodes, ampersand replacement=\&, column sep = 0.2cm, row sep = 0.2cm]{
                {} \& $\phi_2$ \& {} \\
                {} \& {} \& {} \\
                {} \& {} \& {} \\
                {} \& {} \& {} \\
                {} \& $\phi_1$ \& {} \\
            };
            \draw (m-1-2) -- (m-2-2.center) ;
            \draw (m-5-2) -- (m-4-2.center) ;
            \draw (m-3-1) -- (m-3-3) ;
            \draw (m-2-2.center) to[out=0,in=0, distance=0.70cm] (m-4-2.center) ;
            \draw (m-2-2.center) to[out=180,in=180, distance=0.70cm] (m-4-2.center) ;
        \end{tikzpicture}
    \end{equation}
    The recursion gives us
    \begin{equation}
        \sfD^\circ=\sfP_0\circ \sfD_\rmint\circ(\,\Iint\circ\Iint\circ\Iint\circ\sfU\circ\sfU~+~\Iint\circ\Iint\circ\sfU\circ\Iint\circ\sfU\,)+\ldots~,
    \end{equation}
    where $\ldots$ stands for terms not contributing to this Feynman diagram. We end up with eight non-vanishing HPL diagrams:
    \begin{subequations}\label{hpl_diag:ex1}
        \begin{equation}
            \begin{tikzpicture}[baseline={([yshift=-.5ex]current bounding box.center)}]
                \matrix (m) [matrix of nodes, ampersand replacement=\&, column sep = 0.1cm, row sep = 0.2cm]{
                    {} \& $\phi_2$ \& {} \& {} \& {} \& {} \& \\
                    {} \& {} \& {} \& {} \& {} \& {} \& \\
                    {} \& {} \& {} \& {} \& {} \& {} \& \\
                    {} \& {} \& {} \& {} \& {} \& {} \& \\
                    {} \& {} \& {} \& {} \& {} \& {} \& \\
                    {} \& {} \& {} \& {} \& {} \& {} \& \\
                    {} \& {} \& {} \& {} \& {} \& {} \& \\
                    {} \& {} \& {} \& {} \& {} \& {} \& \\
                    {} \& {} \& {} \& {} \& {} \& {} \& \\
                    {} \& {} \& {} \& {} \& {} \& {} \& \\
                    {} \& {} \& {} \& {} \& {} \& {} \& \\
                    {} \& {} \& {} \& {} \& {} \& {} $\phi_1$ \\
                };
                \draw[anti] (m-1-2) -- (m-2-2.center) ;
                \draw[int] (m-2-2.center) -- (m-3-1) ;
                \draw[int] (m-2-2.center) -- (m-3-3) ;
                \draw[int] (m-3-3) -- (m-4-3.center);
                \draw[int] (m-4-3.center) -- (m-5-2) ;
                \draw[int] (m-4-3.center) -- (m-5-4) ;
                \draw[int] (m-5-4) -- (m-6-4.center);
                \draw[int] (m-6-4.center) -- (m-7-3) ;
                \draw[int] (m-6-4.center) -- (m-7-5) ;
                \draw[int] (m-3-1) -- (m-5-1);
                \draw[int] (m-5-1) -- (m-7-1);
                \draw[int] (m-7-1) -- (m-9-1);
                \draw[int] (m-5-2) -- (m-7-2);
                \draw[int] (m-7-2) -- (m-9-2);
                \draw[int] (m-7-3) -- (m-9-3);
                \draw[int] (m-7-5) -- (m-8-5.center);
                \draw[int] (m-8-5.center) -- (m-9-4) ;
                \draw (m-8-5.center) -- (m-9-6) ;
                \draw[int] (m-9-1) -- (m-11-1.center);
                \draw[int] (m-9-3) -- (m-11-3.center);
                \draw[int] (m-11-1.center) -- (m-11-3.center);
                \draw[int] (m-9-2) -- (m-10-2.center);
                \draw[int] (m-9-4) -- (m-10-4.center);
                \draw[int] (m-10-2.center) -- (m-10-4.center);
                \draw (m-9-6) -- (m-12-6);
            \end{tikzpicture}
            ~~~
            \begin{tikzpicture}[baseline={([yshift=-.5ex]current bounding box.center)}]
                \matrix (m) [matrix of nodes, ampersand replacement=\&, column sep = 0.1cm, row sep = 0.2cm]{
                    {} \& $\phi_2$ \& {} \& {} \& {} \& {} \& \\
                    {} \& {} \& {} \& {} \& {} \& {} \& \\
                    {} \& {} \& {} \& {} \& {} \& {} \& \\
                    {} \& {} \& {} \& {} \& {} \& {} \& \\
                    {} \& {} \& {} \& {} \& {} \& {} \& \\
                    {} \& {} \& {} \& {} \& {} \& {} \& \\
                    {} \& {} \& {} \& {} \& {} \& {} \& \\
                    {} \& {} \& {} \& {} \& {} \& {} \& \\
                    {} \& {} \& {} \& {} \& {} \& {} \& \\
                    {} \& {} \& {} \& {} \& {} \& {} \& \\
                    {} \& {} \& {} \& {} \& {} \& {} \& \\
                    {} \& {} \& {} \& $\phi_1$ \& {} \& {} \\
                };
                \draw[anti] (m-1-2) -- (m-2-2.center) ;
                \draw[int] (m-2-2.center) -- (m-3-1) ;
                \draw[int] (m-2-2.center) -- (m-3-3) ;
                \draw[int] (m-3-3) -- (m-4-3.center);
                \draw[int] (m-4-3.center) -- (m-5-2) ;
                \draw[int] (m-4-3.center) -- (m-5-4) ;
                \draw[int] (m-5-4) -- (m-6-4.center);
                \draw[int] (m-6-4.center) -- (m-7-3) ;
                \draw[int] (m-6-4.center) -- (m-7-5) ;
                \draw[int] (m-3-1) -- (m-5-1);
                \draw[int] (m-5-1) -- (m-7-1);
                \draw[int] (m-7-1) -- (m-9-1);
                \draw[int] (m-5-2) -- (m-7-2);
                \draw[int] (m-7-2) -- (m-9-2);
                \draw[int] (m-7-3) -- (m-9-3);
                \draw[int] (m-7-5) -- (m-8-5.center);
                \draw (m-8-5.center) -- (m-9-4) ;
                \draw[int] (m-8-5.center) -- (m-9-6) ;
                \draw[int] (m-9-1) -- (m-11-1.center);
                \draw[int] (m-9-3) -- (m-11-3.center);
                \draw[int] (m-11-1.center) -- (m-11-3.center);
                \draw[int] (m-9-2) -- (m-10-2.center);
                \draw[int] (m-9-6) -- (m-10-6.center);
                \draw[int] (m-10-2.center) -- (m-10-6.center);
                \draw (m-9-4) -- (m-12-4);
            \end{tikzpicture}
            ~~~
            \begin{tikzpicture}[baseline={([yshift=-.5ex]current bounding box.center)}]
                \matrix (m) [matrix of nodes, ampersand replacement=\&, column sep = 0.1cm, row sep = 0.2cm]{
                    {} \& $\phi_2$ \& {} \& {} \& {} \& {} \& \\
                    {} \& {} \& {} \& {} \& {} \& {} \& \\
                    {} \& {} \& {} \& {} \& {} \& {} \& \\
                    {} \& {} \& {} \& {} \& {} \& {} \& \\
                    {} \& {} \& {} \& {} \& {} \& {} \& \\
                    {} \& {} \& {} \& {} \& {} \& {} \& \\
                    {} \& {} \& {} \& {} \& {} \& {} \& \\
                    {} \& {} \& {} \& {} \& {} \& {} \& \\
                    {} \& {} \& {} \& {} \& {} \& {} \& \\
                    {} \& {} \& {} \& {} \& {} \& {} \& \\
                    {} \& {} \& {} \& {} \& {} \& {} \& \\
                    {} \& {} \& {} \& {} \& {} $\phi_1$ \& {} \\
                };
                \draw[anti] (m-1-2) -- (m-2-2.center) ;
                \draw[int] (m-2-2.center) -- (m-3-1) ;
                \draw[int] (m-2-2.center) -- (m-3-3) ;
                \draw[int] (m-3-3) -- (m-4-3.center);
                \draw[int] (m-4-3.center) -- (m-5-2) ;
                \draw[int] (m-4-3.center) -- (m-5-5) ;
                \draw[int] (m-5-5) -- (m-6-5.center);
                \draw[int] (m-6-5.center) -- (m-7-4) ;
                \draw[int] (m-6-5.center) -- (m-7-6) ;
                \draw[int] (m-3-1) -- (m-5-1);
                \draw[int] (m-5-1) -- (m-7-1);
                \draw[int] (m-5-2) -- (m-7-2);
                \draw[int] (m-7-2) -- (m-10-2);
                \draw[int] (m-7-4) -- (m-9-4.center);
                \draw[int] (m-9-4.center) -- (m-10-3) ;
                \draw (m-9-4.center) -- (m-10-5) ;
                \draw[int] (m-10-2) -- (m-11-2.center);
                \draw[int] (m-10-3) -- (m-11-3.center);
                \draw[int] (m-11-2.center) -- (m-11-3.center);
                \draw[int] (m-7-1) -- (m-8-1.center);
                \draw[int] (m-7-6) -- (m-8-6.center);
                \draw[int] (m-8-1.center) -- (m-8-6.center);
                \draw (m-10-5) -- (m-12-5);
            \end{tikzpicture}
            ~~~
            \begin{tikzpicture}[baseline={([yshift=-.5ex]current bounding box.center)}]
                \matrix (m) [matrix of nodes, ampersand replacement=\&, column sep = 0.1cm, row sep = 0.2cm]{
                    {} \& $\phi_2$ \& {} \& {} \& {} \& {} \& \\
                    {} \& {} \& {} \& {} \& {} \& {} \& \\
                    {} \& {} \& {} \& {} \& {} \& {} \& \\
                    {} \& {} \& {} \& {} \& {} \& {} \& \\
                    {} \& {} \& {} \& {} \& {} \& {} \& \\
                    {} \& {} \& {} \& {} \& {} \& {} \& \\
                    {} \& {} \& {} \& {} \& {} \& {} \& \\
                    {} \& {} \& {} \& {} \& {} \& {} \& \\
                    {} \& {} \& {} \& {} \& {} \& {} \& \\
                    {} \& {} \& {} \& {} \& {} \& {} \& \\
                    {} \& {} \& {} \& {} \& {} \& {} \& \\
                    {} \& {} \& {}  $\phi_1$ \& {} \& {}\& {} \\
                };
                \draw[anti] (m-1-2) -- (m-2-2.center) ;
                \draw[int] (m-2-2.center) -- (m-3-1) ;
                \draw[int] (m-2-2.center) -- (m-3-3) ;
                \draw[int] (m-3-3) -- (m-4-3.center);
                \draw[int] (m-4-3.center) -- (m-5-2) ;
                \draw[int] (m-4-3.center) -- (m-5-5) ;
                \draw[int] (m-5-5) -- (m-6-5.center);
                \draw[int] (m-6-5.center) -- (m-7-4) ;
                \draw[int] (m-6-5.center) -- (m-7-6) ;
                \draw[int] (m-3-1) -- (m-5-1);
                \draw[int] (m-5-1) -- (m-7-1);
                \draw[int] (m-5-2) -- (m-7-2);
                \draw[int] (m-7-2) -- (m-10-2);
                \draw[int] (m-7-4) -- (m-9-4.center);
                \draw (m-9-4.center) -- (m-10-3) ;
                \draw[int] (m-9-4.center) -- (m-10-5) ;
                \draw[int] (m-10-2) -- (m-11-2.center);
                \draw[int] (m-10-5) -- (m-11-5.center);
                \draw[int] (m-11-2.center) -- (m-11-5.center);
                \draw[int] (m-7-1) -- (m-8-1.center);
                \draw[int] (m-7-6) -- (m-8-6.center);
                \draw[int] (m-8-1.center) -- (m-8-6.center);
                \draw (m-10-3) -- (m-12-3);
            \end{tikzpicture}
        \end{equation}
        \begin{equation}
            \\
            \begin{tikzpicture}[baseline={([yshift=-.5ex]current bounding box.center)}]
                \matrix (m) [matrix of nodes, ampersand replacement=\&, column sep = 0.1cm, row sep = 0.2cm]{
                    {} \& $\phi_2$ \& {} \& {} \& {} \& {} \& \\
                    {} \& {} \& {} \& {} \& {} \& {} \& \\
                    {} \& {} \& {} \& {} \& {} \& {} \& \\
                    {} \& {} \& {} \& {} \& {} \& {} \& \\
                    {} \& {} \& {} \& {} \& {} \& {} \& \\
                    {} \& {} \& {} \& {} \& {} \& {} \& \\
                    {} \& {} \& {} \& {} \& {} \& {} \& \\
                    {} \& {} \& {} \& {} \& {} \& {} \& \\
                    {} \& {} \& {} \& {} \& {} \& {} \& \\
                    {} \& {} \& {} \& {} \& {} \& {} \& \\
                    {} \& {} \& {} \& {} \& {} \& {} \& \\
                    {} \& {} \& {} \& {} \& {} \& {} $\phi_1$ \\
                };
                \draw[anti] (m-1-2) -- (m-2-2.center) ;
                \draw[int] (m-2-2.center) -- (m-3-1) ;
                \draw[int] (m-2-2.center) -- (m-3-3) ;
                \draw[int] (m-3-3) -- (m-4-3.center);
                \draw[int] (m-4-3.center) -- (m-5-2) ;
                \draw[int] (m-4-3.center) -- (m-5-5) ;
                \draw[int] (m-5-5) -- (m-6-5.center);
                \draw[int] (m-6-5.center) -- (m-7-4) ;
                \draw (m-6-5.center) -- (m-7-6) ;
                \draw[int] (m-3-1) -- (m-5-1);
                \draw[int] (m-5-1) -- (m-7-1);
                \draw[int] (m-7-1) -- (m-9-1);
                \draw[int] (m-5-2) -- (m-7-2);
                \draw[int] (m-7-2) -- (m-9-2);
                \draw[int] (m-7-4) -- (m-8-4.center);
                \draw[int] (m-8-4.center) -- (m-9-3) ;
                \draw[int] (m-8-4.center) -- (m-9-5) ;
                \draw[int] (m-9-1) -- (m-11-1.center);
                \draw[int] (m-9-3) -- (m-11-3.center);
                \draw[int] (m-11-1.center) -- (m-11-3.center);
                \draw[int] (m-9-2) -- (m-10-2.center);
                \draw[int] (m-9-5) -- (m-10-5.center);
                \draw[int] (m-10-2.center) -- (m-10-5.center);
                \draw (m-7-6) -- (m-9-6);
                \draw (m-9-6) -- (m-12-6);
            \end{tikzpicture}
            ~~~
            \begin{tikzpicture}[baseline={([yshift=-.5ex]current bounding box.center)}]
                \matrix (m) [matrix of nodes, ampersand replacement=\&, column sep = 0.1cm, row sep = 0.2cm]{
                    {} \& $\phi_2$ \& {} \& {} \& {} \& {} \& \\
                    {} \& {} \& {} \& {} \& {} \& {} \& \\
                    {} \& {} \& {} \& {} \& {} \& {} \& \\
                    {} \& {} \& {} \& {} \& {} \& {} \& \\
                    {} \& {} \& {} \& {} \& {} \& {} \& \\
                    {} \& {} \& {} \& {} \& {} \& {} \& \\
                    {} \& {} \& {} \& {} \& {} \& {} \& \\
                    {} \& {} \& {} \& {} \& {} \& {} \& \\
                    {} \& {} \& {} \& {} \& {} \& {} \& \\
                    {} \& {} \& {} \& {} \& {} \& {} \& \\
                    {} \& {} \& {} \& {} \& {} \& {} \& \\
                    {} \& {} \& {} \& {} \& {} \& {} $\phi_1$ \\
                };
                \draw[anti] (m-1-2) -- (m-2-2.center) ;
                \draw[int] (m-2-2.center) -- (m-3-1) ;
                \draw[int] (m-2-2.center) -- (m-3-3) ;
                \draw[int] (m-3-3) -- (m-4-3.center);
                \draw[int] (m-4-3.center) -- (m-5-2) ;
                \draw[int] (m-4-3.center) -- (m-5-5) ;
                \draw[int] (m-5-5) -- (m-6-5.center);
                \draw[int] (m-6-5.center) -- (m-7-4) ;
                \draw (m-6-5.center) -- (m-7-6) ;
                \draw[int] (m-3-1) -- (m-5-1);
                \draw[int] (m-5-1) -- (m-7-1);
                \draw[int] (m-7-1) -- (m-9-1);
                \draw[int] (m-5-2) -- (m-7-2);
                \draw[int] (m-7-2) -- (m-9-2);
                \draw[int] (m-7-4) -- (m-8-4.center);
                \draw[int] (m-8-4.center) -- (m-9-3) ;
                \draw[int] (m-8-4.center) -- (m-9-5) ;
                \draw[int] (m-9-2) -- (m-11-2.center);
                \draw[int] (m-9-3) -- (m-11-3.center);
                \draw[int] (m-11-2.center) -- (m-11-3.center);
                \draw[int] (m-9-1) -- (m-10-1.center);
                \draw[int] (m-9-5) -- (m-10-5.center);
                \draw[int] (m-10-1.center) -- (m-10-5.center);
                \draw (m-7-6) -- (m-9-6);
                \draw (m-9-6) -- (m-12-6);
            \end{tikzpicture}
            ~~~
            \begin{tikzpicture}[baseline={([yshift=-.5ex]current bounding box.center)}]
                \matrix (m) [matrix of nodes, ampersand replacement=\&, column sep = 0.1cm, row sep = 0.2cm]{
                    {} \& $\phi_2$ \& {} \& {} \& {} \& {} \& \\
                    {} \& {} \& {} \& {} \& {} \& {} \& \\
                    {} \& {} \& {} \& {} \& {} \& {} \& \\
                    {} \& {} \& {} \& {} \& {} \& {} \& \\
                    {} \& {} \& {} \& {} \& {} \& {} \& \\
                    {} \& {} \& {} \& {} \& {} \& {} \& \\
                    {} \& {} \& {} \& {} \& {} \& {} \& \\
                    {} \& {} \& {} \& {} \& {} \& {} \& \\
                    {} \& {} \& {} \& {} \& {} \& {} \& \\
                    {} \& {} \& {} \& {} \& {} \& {} \& \\
                    {} \& {} \& {} \& {} \& {} \& {} \& \\
                    {} \& {} \& $\phi_1$ \& {} \& {} \& {} \& {} \\
                };
                \draw[anti] (m-1-2) -- (m-2-2.center) ;
                \draw[int] (m-2-2.center) -- (m-3-1) ;
                \draw[int] (m-2-2.center) -- (m-3-3) ;
                \draw[int] (m-3-3) -- (m-4-3.center);
                \draw[int] (m-4-3.center) -- (m-5-2) ;
                \draw[int] (m-4-3.center) -- (m-5-4) ;
                \draw[int] (m-5-4) -- (m-6-4.center);
                \draw (m-6-4.center) -- (m-7-3) ;
                \draw[int] (m-6-4.center) -- (m-7-5) ;
                \draw[int] (m-3-1) -- (m-5-1);
                \draw[int] (m-5-1) -- (m-7-1);
                \draw[int] (m-7-1) -- (m-9-1);
                \draw[int] (m-5-2) -- (m-7-2);
                \draw[int] (m-7-2) -- (m-9-2);
                \draw[int] (m-7-5) -- (m-8-5.center);
                \draw[int] (m-8-5.center) -- (m-9-4) ;
                \draw[int] (m-8-5.center) -- (m-9-6) ;
                \draw[int] (m-9-1) -- (m-11-1.center);
                \draw[int] (m-9-4) -- (m-11-4.center);
                \draw[int] (m-11-1.center) -- (m-11-4.center);
                \draw[int] (m-9-2) -- (m-10-2.center);
                \draw[int] (m-9-6) -- (m-10-6.center);
                \draw[int] (m-10-2.center) -- (m-10-6.center);
                \draw (m-7-3) -- (m-9-3);
                \draw (m-9-3) -- (m-12-3);
            \end{tikzpicture}
            ~~~
            \begin{tikzpicture}[baseline={([yshift=-.5ex]current bounding box.center)}]
                \matrix (m) [matrix of nodes, ampersand replacement=\&, column sep = 0.1cm, row sep = 0.2cm]{
                    {} \& $\phi_2$ \& {} \& {} \& {} \& {} \& \\
                    {} \& {} \& {} \& {} \& {} \& {} \& \\
                    {} \& {} \& {} \& {} \& {} \& {} \& \\
                    {} \& {} \& {} \& {} \& {} \& {} \& \\
                    {} \& {} \& {} \& {} \& {} \& {} \& \\
                    {} \& {} \& {} \& {} \& {} \& {} \& \\
                    {} \& {} \& {} \& {} \& {} \& {} \& \\
                    {} \& {} \& {} \& {} \& {} \& {} \& \\
                    {} \& {} \& {} \& {} \& {} \& {} \& \\
                    {} \& {} \& {} \& {} \& {} \& {} \& \\
                    {} \& {} \& {} \& {} \& {} \& {} \& \\
                    {} \& {} \& $\phi_1$ \& {} \& {} \& {} \& {} \\
                };
                \draw[anti] (m-1-2) -- (m-2-2.center) ;
                \draw[int] (m-2-2.center) -- (m-3-1) ;
                \draw[int] (m-2-2.center) -- (m-3-3) ;
                \draw[int] (m-3-3) -- (m-4-3.center);
                \draw[int] (m-4-3.center) -- (m-5-2) ;
                \draw[int] (m-4-3.center) -- (m-5-4) ;
                \draw[int] (m-5-4) -- (m-6-4.center);
                \draw (m-6-4.center) -- (m-7-3) ;
                \draw[int] (m-6-4.center) -- (m-7-5) ;
                \draw[int] (m-3-1) -- (m-5-1);
                \draw[int] (m-5-1) -- (m-7-1);
                \draw[int] (m-7-1) -- (m-9-1);
                \draw[int] (m-5-2) -- (m-7-2);
                \draw[int] (m-7-2) -- (m-9-2);
                \draw[int] (m-7-5) -- (m-8-5.center);
                \draw[int] (m-8-5.center) -- (m-9-4) ;
                \draw[int] (m-8-5.center) -- (m-9-6) ;
                \draw[int] (m-9-2) -- (m-11-2.center);
                \draw[int] (m-9-4) -- (m-11-4.center);
                \draw[int] (m-11-2.center) -- (m-11-4.center);
                \draw[int] (m-9-1) -- (m-10-1.center);
                \draw[int] (m-9-6) -- (m-10-6.center);
                \draw[int] (m-10-1.center) -- (m-10-6.center);
                \draw (m-7-3) -- (m-9-3);
                \draw (m-9-3) -- (m-12-3);
            \end{tikzpicture}
        \end{equation}
    \end{subequations}
    where we did not connect lines to clearly indicate the action of the operators $\Iint$ and $\sfU$. All other HPL diagrams are either of different topology to~\eqref{diag:ex_1} or vanish, mostly due to an operator $\sfe\circ \sfp$ hitting a dotted line.

    \section{Symmetry factors of Feynman diagrams}
    
    \subsection{Generating functional}\label{ssec:gen_func}
    
    Recall the generating functional for connected $n$-point correlation functions in real scalar field theory,
    \begin{equation}\label{eq:generating_function}
        W[J]\coloneqq  \rme^{\rmi \int \rmd^d w\caL_{\rm int}[\frac{1}{\rmi}\delder{J(w)}]}~\rme^{\frac{\rmi}{2}\int \rmd^d y\rmd^d z J(y)G_2(y-z)J(z)}~,
    \end{equation}
    from which we extract the correlation functions
    \begin{equation}\label{eq:correlators}
        \langle \phi(x_1)\ldots \phi(x_{n})\rangle=\left.\delder{J(x_1)}\ldots \delder{J(x_{n})}~W[J]\right|_{J=0}
    \end{equation}
    after Taylor expanding the exponentials in $W[J]$. We normalize the coupling constants $c_k$ in the interaction Lagrangian as 
    \begin{equation}
        \caL_{\rm int}[\phi(x)]=\sum_{k\geq 3} \frac{1}{k!}c_k \phi^k(x)~,
    \end{equation}
    cf.~\eqref{eq:action}, which will allow us to extract the symmetry factor $\Sigma_\Gamma$ of a Feynman diagram $\Gamma$ in the straightforward, usual manner presented in quantum field theory textbooks.
    
    We note that even though the Feynman diagrams contributing to scattering amplitudes have amputated external legs, we will always consider un-amputated diagrams to compute the symmetry factors from the generating functional for the corresponding amputated ones.
    
    For a Feynman diagram $\Gamma$ with $v_i$ $i$-ary vertices and $p$ propagators (no amputations of external legs), the Taylor expansion leads to the following denominator $D_\Gamma$ of the symmetry factor:
    \begin{equation}
        \frac{1}{D_\Gamma}=\frac{1}{p!}\left(\frac{1}{2}\right)^p\prod_i \frac{1}{v_i!}\left(\frac{1}{i!}\right)^{v_i}~.
    \end{equation}
    The numerator $N_\Gamma$ now counts the number of distinct actions of the functional derivatives on the Taylor-expanded exponentials to produce the Feynman diagram under consideration. As is well known, the total symmetry factor is therefore the quotient by the automorphism group $\mathsf{Aut}(\Gamma)$ of the Feynman diagram:
    \begin{equation}
        \Sigma_\Gamma=\frac{N_\Gamma}{D_\Gamma}=\frac{1}{|\mathsf{Aut}(\Gamma)|}~.
    \end{equation}
    
    A general and convenient formula for $\Sigma_\Gamma$ for general field theories with cubic and quartic potentials has been given in~\cite{Palmer:2001vq}. We now extend it to the general real scalar field theory with arbitrary polynomial potential. Non-trivial symmetry factors arise as follows:
    \begin{enumerate}
        \item A factor of $\frac{1}{\pi_v}$, where $\pi_v$ is the number of permutations of internal vertices so that all propagators remain connected to the same propagators. For example, the two vertices in the middle of the diagram~\eqref{diag:ex_1} can be permuted and thus $\pi_v=2!$.
        \item A factor of $\frac{1}{\ell_i! 2^{\ell_i}}$ for any $\ell_i$ propagators connecting to the same vertex, i.e.~subdiagrams of the form
        \begin{equation}\label{diag:self_loops}
            \begin{tikzpicture}[baseline={([yshift=-.5ex]current bounding box.center)}]
                \matrix (m) [matrix of nodes, ampersand replacement=\&, column sep = 0.2cm, row sep = 0.2cm]{
                    {} \& $\vdots$ \& {} \\
                    {} \& {} \& {} \\
                    {} \& $\ldots$ \& {} \\
                    {} \& {} \& {} \\
                };
                \draw (m-3-1) -- (m-2-2.center) ;
                \draw (m-3-3) -- (m-2-2.center) ;
                \draw (m-2-2.center) arc(-90:270:0.15);
                \draw (m-2-2.center) arc(-90:270:0.65);
                \draw (m-2-2.center) arc(-90:270:0.55);
            \end{tikzpicture}
        \end{equation}
        This is easily seen: the numerator from the Taylor expansion for $k$ external legs is
        \begin{equation}
            N_\Gamma=\frac{1}{(k+\ell_i)!}\frac{1}{2^{k+\ell_i}}\frac{1}{(k+2\ell_i)!}~.
        \end{equation}
        We have $\frac{(k+\ell_i)!}{\ell_i!}2^k$ ways of connecting the external vertices to propagators. There are then $\binom{k+2\ell_i}{k}k!$ ways of connecting the vertex to these propagators. It remains to connect the $2\ell_i$ derivatives left over from the vertex with the $2\ell_i$ fields of the remaining $\ell_i$ propagators, which can be done in $2\ell_i!$ ways.      Putting everything together, we are left with the given symmetry factor.
        \item A factor of $\frac{1}{n_j!}$ for any $n_j$ propagators between the same two distinct vertices, i.e.~subdiagrams of the form
        \begin{equation}\label{diag:2-loops}
            \begin{tikzpicture}[baseline={([yshift=-.5ex]current bounding box.center)}]
                \matrix (m) [matrix of nodes, ampersand replacement=\&, column sep = 0.2cm, row sep = 0.2cm]{
                    {} \& $\ldots$ \& {} \\
                    {} \& {} \& {} \\
                    {} \& $~\ldots$ \& {} \\
                    {} \& {} \& {} \\
                    {} \& $\ldots$ \& {} \\
                };
                \draw (m-1-1) -- (m-2-2.center) ;
                \draw (m-1-3) -- (m-2-2.center) ;
                \draw (m-5-1) -- (m-4-2.center) ;
                \draw (m-5-3) -- (m-4-2.center) ;
                \draw (m-2-2.center) to[out=0,in=0, distance=0.70cm] (m-4-2.center) ;
                \draw (m-2-2.center) to[out=180,in=180, distance=0.50cm] (m-4-2.center) ;
                \draw (m-2-2.center) to[out=180,in=180, distance=0.70cm] (m-4-2.center) ;
            \end{tikzpicture}
        \end{equation}
        Let us assume that there are $k$ external legs and two vertices of degree~$d_1$ and $d_2=k+2\ell-d_1$. We further assume that $d_1\neq d_2$ so that the vertices are distinguishable. Otherwise, there is an additional factor of $\tfrac12$ in the numerator which is compensated by a factor of~$2$ coming from permuting the two vertices. The numerator is then 
        \begin{equation}
            N_\Gamma=\frac{1}{(k+n_j)!}\frac{1}{2^{k+n_j}}\frac{1}{d_1!}\frac{1}{d_2!}~.
        \end{equation}
        There are again $\frac{(k+n_j)!}{n_j!}2^{k}$ ways of connecting the external vertices to propagators. Note that $k_1=d_1-n_j$ and $k_2=d_2-n_j$ of the external legs are connected to vertex~1 and~2, respectively. There are $\binom{d_1}{k_1}k_1!\binom{d_2}{k_2}k_2!$ ways of doing this. Finally, we can connect the first vertex with the remaining $n_j$ propagators, which can be done in $n_j!2^{n_j}$ possible ways, and the second vertex to them which is done in $n_j!$ ways. We are left with a total symmetry factor of $\frac{1}{n_j!}$.
    \end{enumerate}
    Altogether, we have the following formula:
    \begin{equation}\label{eq:formula_Sigma}
        \Sigma_\Gamma=\frac{1}{\pi_v}\left(\prod_i\frac{1}{\ell_i! 2^{\ell_i}}\right)\left(\prod_j \frac{1}{n_j!}\right)~.
    \end{equation}
    
    As a corollary, we have the familiar textbook statement:
    \begin{corollary}
        All diagrams $\Gamma$ contributing to tree level amplitudes have a symmetry factor of $\Sigma_\Gamma=1$.
    \end{corollary}
    \begin{proof}
        All vertices and propagators are distinguished by their position relative to the labeled output legs. Therefore, there are no two different permutations of vertices and propagators leading to the same diagram. This means $\pi_v=1$ and all the factors in the Taylor expansion of the generating function~\eqref{eq:generating_function} are absorbed. More abstractly, the automorphism group of the Feynman diagram is trivial. 
    \end{proof}
    
    It will be useful to have a concrete algorithm for evaluating the contributions of the correlators~\eqref{eq:correlators} to a concrete Feynman diagram.
    \begin{algorithm}\label{alg:main}
        To evaluate a given Feynman diagram contributing to the correlators~\eqref{eq:correlators}, we proceed as follows:
        \begin{enumerate}
            \item We start from the functional derivative with respect to $J(x_{n+1})$, connect it to a propagator and to a suitable vertex. We call the external leg corresponding to $\phi(x_{n+1})$ the {\em origin} of our graph.
            \item We pair off any legs leading to external legs in the vertices  connected so far to the origin.
            \item We close off all loops between vertices connected so far to the origin.
            \item If there are no vertices left, we are done. Otherwise, we take the vertex with remaining unpaired functional derivatives which we added last. We then pair this vertex with a propagator and a remaining vertex. (Note that all legs leading to loops and exernal legs have already been paired off.) We then continue with 2.
        \end{enumerate}
    \end{algorithm}
    
    Because any (un-amputated) Feynman diagram contributes to a correlator, it is clear that the above algorithm will be suitable for any Feynman diagram.
    
    \subsection{Homological perturbation lemma: examples}
    
    We want to show that the homological perturbation lemma reproduces the combinatorics of scattering amplitudes. That is, each Feynman diagram $\Gamma$ is produced by the homological perturbation lemma with the appropriate symmetry factor $\Sigma_\Gamma$. Before stating the general proof, let us discuss a few illustrative examples.
    
    As a first, non-trivial example, consider again the diagram~\eqref{diag:ex_1}, which has a symmetry factor of $\Sigma_\Gamma=\tfrac12$. The HPL reproduces this diagram in eight possible ways, cf.~the HPL-diagrams~\eqref{hpl_diag:ex1}. Each of the four vertices comes with a factor of\footnote{We shall always drop the coupling constants $c_k$ when talking about the factors; they are evidently reproduced in the correct way.} $\tfrac12$, and the total factor we obtain is thus $8\frac{1}{2^4}$, as expected.
    
    As a second example, we consider the diagram
    \begin{equation}\label{diag:4pt1lp}
        \begin{tikzpicture}[baseline={([yshift=-.5ex]current bounding box.center)}]
            \matrix (m) [matrix of nodes, ampersand replacement=\&, column sep = 0.2cm, row sep = 0.2cm]{
                {} \& $4$ \& {} \\
                {} \& {} \& {} \\
                {} \& {} \& {} \\
                {} \& {} \& {} \\
                $1$ \& $2$ \& $3$ \\
            };
            \draw (m-1-2) -- (m-2-2.center) ;
            \draw (m-5-2) -- (m-4-2.center) ;
            \draw (m-2-2.center) to[out=0,in=0, distance=0.69cm] node[midway,inner sep=0pt,outer sep=0pt](NP1) {} (m-4-2.center) ;
            \draw (m-2-2.center) to[out=180,in=180, distance=0.69cm] node[midway,inner sep=0pt,outer sep=0pt](NP2) {} (m-4-2.center);
            \draw (NP2) -- (m-5-1) ;
            \draw (NP1) -- (m-5-3) ;
        \end{tikzpicture}
        ~~=~~
        \begin{tikzpicture}[baseline={([yshift=-.5ex]current bounding box.center)}]
            \matrix (m) [matrix of nodes, ampersand replacement=\&, column sep = 0.2cm, row sep = 0.2cm]{
                {} \& $4$ \& {} \\
                {} \& {} \& {} \\
                {} \& {} \& {} \\
                {} \& {} \& {} \\
                $3$ \& $2$ \& $1$ \\
            };
            \draw (m-1-2) -- (m-2-2.center) ;
            \draw (m-5-2) -- (m-4-2.center) ;
            \draw (m-2-2.center) to[out=0,in=0, distance=0.69cm] node[midway,inner sep=0pt,outer sep=0pt](NP1) {} (m-4-2.center) ;
            \draw (m-2-2.center) to[out=180,in=180, distance=0.69cm] node[midway,inner sep=0pt,outer sep=0pt](NP2) {} (m-4-2.center);
            \draw (NP2) -- (m-5-1) ;
            \draw (NP1) -- (m-5-3) ;
        \end{tikzpicture}
    \end{equation}
    which comes with a symmetry factor of $\Sigma_\Gamma=1$. Since it has again four vertices contributing a factor of $(\tfrac12)^4$, the HPL needs to produce this diagram in 16 distinct ways. Consider some of the realizations of this diagram by the homological perturbation lemma: 
    \begin{equation}
        \begin{tikzpicture}[baseline={([yshift=-.5ex]current bounding box.center)}]
            \matrix (m) [matrix of nodes, ampersand replacement=\&, column sep = 0.2cm, row sep = 0.2cm]{
                {} \& $4$ \& {} \& {} \& {} \& {} \& \\
                {} \& {} \& {} \& {} \& {} \& {} \& \\
                {} \& {} \& {} \& {} \& {} \& {} \& \\
                {} \& {} \& {} \& {} \& {} \& {} \& \\
                {} \& {} \& {} \& {} \& {} \& {} \& \\
                {} \& {} \& {} \& {} \& {} \& {} \& \\
                {} \& {} \& {} \& {} \& {} \& {} \& \\
                {} \& {} \& {} \& {} \& {} \& {} \& \\
                {} \& {} \& {} \& {} \& {} \& {} \& \\
                {} \& {} \& {} \& {} \& {} \& {} \& \\
                {} \& $1$ \& $2$ \& $3$ \& {} \& {} \& {} \\
            };
            \draw[anti] (m-1-2) -- (m-2-2.center) ;
            \draw[int] (m-2-2.center) -- (m-3-1) ;
            \draw[int] (m-2-2.center) -- (m-3-3) ;
            \draw[int] (m-3-3) -- (m-4-3.center);
            \draw (m-4-3.center) -- (m-5-2) ;
            \draw[int] (m-4-3.center) -- (m-5-4) ;
            \draw[int] (m-5-4) -- (m-6-4.center);
            \draw (m-6-4.center) -- (m-7-3) ;
            \draw[int] (m-6-4.center) -- (m-7-5) ;
            \draw[int] (m-3-1) -- (m-5-1);
            \draw[int] (m-5-1) -- (m-7-1);
            \draw[int] (m-7-1) -- (m-9-1);
            \draw (m-5-2) -- (m-7-2);
            \draw (m-7-2) -- (m-9-2);
            \draw (m-9-2) -- (m-11-2);
            \draw (m-9-4) -- (m-11-4);
            \draw[int] (m-7-5) -- (m-8-5.center);
            \draw (m-8-5.center) -- (m-9-4) ;
            \draw[int] (m-8-5.center) -- (m-9-6) ;
            \draw[int] (m-9-1) -- (m-10-1.center);
            \draw[int] (m-9-6) -- (m-10-6.center);
            \draw[int] (m-10-1.center) -- (m-10-6.center);
            \draw (m-7-3) -- (m-9-3);
            \draw (m-9-3) -- (m-11-3);
        \end{tikzpicture}
        \begin{tikzpicture}[baseline={([yshift=-.5ex]current bounding box.center)}]
            \matrix (m) [matrix of nodes, ampersand replacement=\&, column sep = 0.2cm, row sep = 0.2cm]{
                {} \& $4$ \& {} \& {} \& {} \& {} \& \\
                {} \& {} \& {} \& {} \& {} \& {} \& \\
                {} \& {} \& {} \& {} \& {} \& {} \& \\
                {} \& {} \& {} \& {} \& {} \& {} \& \\
                {} \& {} \& {} \& {} \& {} \& {} \& \\
                {} \& {} \& {} \& {} \& {} \& {} \& \\
                {} \& {} \& {} \& {} \& {} \& {} \& \\
                {} \& {} \& {} \& {} \& {} \& {} \& \\
                {} \& {} \& {} \& {} \& {} \& {} \& \\
                {} \& {} \& {} \& {} \& {} \& {} \& \\
                {} \& $3$ \& $2$ \& $1$ \& {} \& {} \& {} \\
            };
            \draw[anti] (m-1-2) -- (m-2-2.center) ;
            \draw[int] (m-2-2.center) -- (m-3-1) ;
            \draw[int] (m-2-2.center) -- (m-3-3) ;
            \draw[int] (m-3-3) -- (m-4-3.center);
            \draw (m-4-3.center) -- (m-5-2) ;
            \draw[int] (m-4-3.center) -- (m-5-4) ;
            \draw[int] (m-5-4) -- (m-6-4.center);
            \draw (m-6-4.center) -- (m-7-3) ;
            \draw[int] (m-6-4.center) -- (m-7-5) ;
            \draw[int] (m-3-1) -- (m-5-1);
            \draw[int] (m-5-1) -- (m-7-1);
            \draw[int] (m-7-1) -- (m-9-1);
            \draw (m-5-2) -- (m-7-2);
            \draw (m-7-2) -- (m-9-2);
            \draw (m-9-2) -- (m-11-2);
            \draw (m-9-4) -- (m-11-4);
            \draw[int] (m-7-5) -- (m-8-5.center);
            \draw (m-8-5.center) -- (m-9-4) ;
            \draw[int] (m-8-5.center) -- (m-9-6) ;
            \draw[int] (m-9-1) -- (m-10-1.center);
            \draw[int] (m-9-6) -- (m-10-6.center);
            \draw[int] (m-10-1.center) -- (m-10-6.center);
            \draw (m-7-3) -- (m-9-3);
            \draw (m-9-3) -- (m-11-3);
        \end{tikzpicture}
        \begin{tikzpicture}[baseline={([yshift=-.5ex]current bounding box.center)}]
            \matrix (m) [matrix of nodes, ampersand replacement=\&, column sep = 0.2cm, row sep = 0.2cm]{
                {} \& $4$ \& {} \& {} \& {} \& {} \& \\
                {} \& {} \& {} \& {} \& {} \& {} \& \\
                {} \& {} \& {} \& {} \& {} \& {} \& \\
                {} \& {} \& {} \& {} \& {} \& {} \& \\
                {} \& {} \& {} \& {} \& {} \& {} \& \\
                {} \& {} \& {} \& {} \& {} \& {} \& \\
                {} \& {} \& {} \& {} \& {} \& {} \& \\
                {} \& {} \& {} \& {} \& {} \& {} \& \\
                {} \& {} \& {} \& {} \& {} \& {} \& \\
                {} \& {} \& {} \& {} \& {} \& {} \& \\
                {} \& $1$ \& $2$ \& {} \& {} \& $3$ \& {} \\
            };
            \draw[anti] (m-1-2) -- (m-2-2.center) ;
            \draw[int] (m-2-2.center) -- (m-3-1) ;
            \draw[int] (m-2-2.center) -- (m-3-3) ;
            \draw[int] (m-3-3) -- (m-4-3.center);
            \draw (m-4-3.center) -- (m-5-2) ;
            \draw[int] (m-4-3.center) -- (m-5-4) ;
            \draw[int] (m-5-4) -- (m-6-4.center);
            \draw (m-6-4.center) -- (m-7-3) ;
            \draw[int] (m-6-4.center) -- (m-7-5) ;
            \draw[int] (m-3-1) -- (m-5-1);
            \draw[int] (m-5-1) -- (m-7-1);
            \draw[int] (m-7-1) -- (m-9-1);
            \draw (m-5-2) -- (m-7-2);
            \draw (m-7-2) -- (m-9-2);
            \draw (m-9-2) -- (m-11-2);
            \draw (m-9-6) -- (m-11-6);
            \draw[int] (m-7-5) -- (m-8-5.center);
            \draw[int] (m-8-5.center) -- (m-9-4) ;
            \draw (m-8-5.center) -- (m-9-6) ;
            \draw[int] (m-9-1) -- (m-10-1.center);
            \draw[int] (m-9-4) -- (m-10-4.center);
            \draw[int] (m-10-1.center) -- (m-10-4.center);
            \draw (m-7-3) -- (m-9-3);
            \draw (m-9-3) -- (m-11-3);
        \end{tikzpicture}
        \begin{tikzpicture}[baseline={([yshift=-.5ex]current bounding box.center)}]
            \matrix (m) [matrix of nodes, ampersand replacement=\&, column sep = 0.2cm, row sep = 0.2cm]{
                {} \& $4$ \& {} \& {} \& {} \& {} \& \\
                {} \& {} \& {} \& {} \& {} \& {} \& \\
                {} \& {} \& {} \& {} \& {} \& {} \& \\
                {} \& {} \& {} \& {} \& {} \& {} \& \\
                {} \& {} \& {} \& {} \& {} \& {} \& \\
                {} \& {} \& {} \& {} \& {} \& {} \& \\
                {} \& {} \& {} \& {} \& {} \& {} \& \\
                {} \& {} \& {} \& {} \& {} \& {} \& \\
                {} \& {} \& {} \& {} \& {} \& {} \& \\
                {} \& {} \& {} \& {} \& {} \& {} \& \\
                {} \& $3$ \& {} \& {} \& $1$ \& $2$ \\
            };
            \draw[anti] (m-1-2) -- (m-2-2.center) ;
            \draw[int] (m-2-2.center) -- (m-3-1) ;
            \draw[int] (m-2-2.center) -- (m-3-3) ;
            \draw[int] (m-3-3) -- (m-4-3.center);
            \draw (m-4-3.center) -- (m-5-2) ;
            \draw[int] (m-4-3.center) -- (m-5-5) ;
            \draw[int] (m-5-5) -- (m-6-5.center);
            \draw[int] (m-6-5.center) -- (m-7-4) ;
            \draw (m-6-5.center) -- (m-7-6) ;
            \draw[int] (m-3-1) -- (m-5-1);
            \draw[int] (m-5-1) -- (m-7-1);
            \draw[int] (m-7-1) -- (m-9-1);
            \draw[int] (m-7-4) -- (m-8-4.center);
            \draw[int] (m-8-4.center) -- (m-9-3) ;
            \draw (m-8-4.center) -- (m-9-5) ;
            \draw[int] (m-9-1) -- (m-10-1.center);
            \draw[int] (m-9-3) -- (m-10-3.center);
            \draw[int] (m-10-1.center) -- (m-10-3.center);
            \draw (m-7-6) -- (m-9-6);
            \draw (m-9-6) -- (m-11-6);
            \draw (m-9-5) -- (m-11-5);
            \draw (m-5-2) -- (m-7-2);
            \draw (m-7-2) -- (m-9-2);
            \draw (m-9-2) -- (m-11-2);
        \end{tikzpicture}
    \end{equation}
    We note that each realization comes with a factor of~2, because we can invert the order of the input legs (cf.~first and second diagram). They come with a further factor of~2 because the lowest cubic vertex can always be flipped (cf.~first and third diagram). In fact, all the vertices except for the very first one can be flipped, one merely has to select the appropriate different (and unique) permutation of the input labels (cf.~fourth diagram). The first vertex cannot be flipped as this would bring a loop line to the right of a vertex which vanishes due to the form of $\sfH$, cf.~\eqref{eq:vanish_diag}. Altogether, we have $2^3$ ways of flipping vertices and $2$ ways of choosing the order of input labels, compensating the factor of $(\tfrac{1}{2})^4$ and giving an overall symmetry factor of $\Sigma_\Gamma=1$.
    
    \begin{lemma}
        The HPL produces the correct symmetry factor $\Sigma_\Gamma$ for the 1-loop diagrams of the form
        \begin{equation}
            \begin{tikzpicture}[baseline={([yshift=-.5ex]current bounding box.center)}]
                \matrix (m) [matrix of nodes, ampersand replacement=\&, column sep = 0.2cm, row sep = 0.2cm]{
                    {} \& $n+1$ \& {} \\
                    {} \& {} \& {} \\
                    {} \& {} \& {} \\
                    {} \& {} \& {} \\
                    $1$ \& $2$ \& $\ldots$ \& $n$\\
                };
                \draw (m-1-2) -- (m-2-2.center) ;
                \draw (m-5-2) -- (m-4-2.center) ;
                \draw (m-2-2.center) to[out=0,in=0, distance=0.69cm] node[midway,inner sep=0pt,outer sep=0pt](NP1) {} (m-4-2.center) ;
                \draw (m-2-2.center) to[out=180,in=180, distance=0.69cm] node[midway,inner sep=0pt,outer sep=0pt](NP2) {} (m-4-2.center);
                \draw (NP2) -- (m-5-1) ;
                \draw (NP1) -- (m-5-4) ;
            \end{tikzpicture}
        \end{equation}
        which is $\Sigma_\Gamma=\tfrac12$ for $n=1$ and $\Sigma_\Gamma=1$ else.
    \end{lemma}
    \begin{proof}
        The case $n>1$ is evident from the argument above for diagram~\eqref{diag:4pt1lp} which readily generalize to the $(n+1)$-point case. For $n=1$, there is just a single way of permuting the single input label, i.e.~half as many ways as for $n>1$, which leads to $\Sigma_\Gamma=\tfrac12$.
    \end{proof}

    \subsection{Homological perturbation lemma: main theorem}
    
    We now come to our main result.
    \begin{theorem}
        The HPL reproduces each Feynman diagram $\Gamma$ with the correct symmetry factor $\Sigma_\Gamma$.
    \end{theorem}
    The proof is based on the fact that each step in the algorithm~\ref{alg:main} has a clear correspondence in the construction of HPL diagrams, in inverse operational order: step one is just the final combination of operators $\sfP\circ \sfD_\rmint$. Step two does not require any operator insertion, while steps three and four correspond to actions of operators $\sfU$ and $\Iint$, respectively. The sensitive issues are around the order of the incoming legs of vertices, which is asymmetric due to~\eqref{eq:vanish_diag}. 

    We begin the proof with the fact that any Feynman diagram is indeed produced by the homological perturbation lemma.
    \begin{lemma}
        The HPL reproduces each Feynman diagram at least once in a non-vanishing way.
    \end{lemma}
    \begin{proof}
        This is evident from algorithm~\ref{alg:main} with the additional prescription\footnote{only temporarily imposed for this proof} in steps 2, 3 and 4 that open incoming legs always have to be filled from right to left. This guarantees that no propagator is produced to the left of an existing line which would cause the diagram to vanish by~\eqref{eq:vanish_diag}.
    \end{proof}

    Next, we present a sequence of lemmata in which each lemma shows that symmetry factors of more complicated Feynman diagrams are reproduced by the HPL, if this is true for simpler diagrams. To streamline the argument, we allow for diagrams with binary vertices as well as tadpole diagrams in our considerations. We end up with tadpole diagrams with no loops involving a single vertex, and a final lemma shows that all their symmetry factors are obtained from the HPL.
    
    We start with tree diagrams:
    \begin{lemma}\label{lem:1}
        The HPL creates each tree diagram $\Gamma$ in the tree level scattering amplitude with a symmetry factor of $\Sigma_\Gamma=1$.
    \end{lemma}
    \begin{proof}
        It is clear that each $(n+1)$-point tree diagram can be turned into a valid HPL diagram by selecting the first $n$ external legs as inputs and the last leg as an output. Note that for the HPL diagram not to be vanishing by identity~\eqref{eq:vanish_diag}, it is enough that the operator $\Iint$ always acts as far to the left as possible to produce the desired diagram. For example, the first two diagrams below contribute, while the last one is vanishing due to~\eqref{eq:vanish_diag}: 
        \begin{equation}\label{eq:tree_examples}
            \begin{tikzpicture}[baseline={([yshift=-.5ex]current bounding box.center)}]
                \matrix (m) [matrix of nodes, ampersand replacement=\&, column sep = 0.03cm, row sep = 0.2cm]{
                    {} \& {} \& {} \& {} \& {} \& {} \& {} \& {} \\
                    {} \& {} \& {} \& {} \& {} \& {} \& {} \& {} \\
                    {} \& {} \& {} \& {} \& {} \& {} \& {} \& {} \\
                    {} \& {} \& {} \& {} \& {} \& {} \& {} \& {} \\
                    {} \& {} \& {} \& {} \& {} \& {} \& {} \& {} \\
                    {} \& {} \& {} \& {} \& {} \& {} \& {} \& {} \\
                    {} \& {} \& {} \& {} \& {} \& {} \& {} \& {} \\
                    {} \& {} \& {} \& {} \& {} \& {} \& {} \& {} \\
                    {} \& {} \& {} \& {} \& {} \& {} \& {} \& {} \\
                    {} \& {} \& {} \& {} \& {} \& {} \& {} \& {} \\
                    $\phi_1$ \& $\phi_2$ \& {} \& $\phi_3$ \& $\phi_4$ \& {} \& $\phi_5$ \\
                };
                \draw[anti] (m-1-3) -- (m-2-3.center) ;
                \draw[int] (m-2-3.center) -- (m-3-2) ;
                \draw[int] (m-2-3.center) -- (m-3-6) ;
                \draw[int] (m-3-2) -- (m-5-2);
                \draw[int] (m-3-6) -- (m-4-6.center);
                \draw (m-4-6.center) -- (m-5-5) ;
                \draw (m-4-6.center) -- (m-5-7) ;
                \draw[int] (m-5-2) -- (m-6-2.center);
                \draw (m-6-2.center) -- (m-7-1) ;
                \draw[int] (m-6-2.center) -- (m-7-3) ;
                \draw[int] (m-7-3) -- (m-8-3.center);
                \draw (m-8-3.center) -- (m-9-2) ;
                \draw (m-8-3.center) -- (m-9-4) ;
                \draw (m-7-1) -- (m-9-1);
                \draw (m-9-1) -- (m-11-1);
                \draw (m-9-2) -- (m-11-2);
                \draw (m-9-4) -- (m-11-4);
                \draw (m-5-5) -- (m-7-5);
                \draw (m-7-5) -- (m-9-5);
                \draw (m-9-5) -- (m-11-5);
                \draw (m-5-7) -- (m-7-7);
                \draw (m-7-7) -- (m-9-7);
                \draw (m-9-7) -- (m-11-7);
            \end{tikzpicture}
            ~~~~~~ 
            \begin{tikzpicture}[baseline={([yshift=-.5ex]current bounding box.center)}]
                \matrix (m) [matrix of nodes, ampersand replacement=\&, column sep = 0.03cm, row sep = 0.2cm]{
                    {} \& {} \& {} \& {} \& {} \& {} \& {}\\
                    {} \& {} \& {} \& {} \& {} \& {} \& {}\\
                    {} \& {} \& {} \& {} \& {} \& {} \& {}\\
                    {} \& {} \& {} \& {} \& {} \& {} \& {}\\
                    {} \& {} \& {} \& {} \& {} \& {} \& {}\\
                    {} \& {} \& {} \& {} \& {} \& {} \& {}\\
                    {} \& {} \& {} \& {} \& {} \& {} \& {}\\
                    {} \& {} \& {} \& {} \& {} \& {} \& {}\\
                    {} \& {} \& {} \& {} \& {} \& {} \& {}\\
                    {} \& {} \& {} \& {} \& {} \& {} \& {}\\
                    $\phi_4$ \& {} \& $\phi_5$ \& $\phi_1$ \& $\phi_2$ \& {} \& $\phi_3$ \\
                };
                \draw[anti] (m-1-3) -- (m-2-3.center) ;
                \draw[int] (m-2-3.center) -- (m-3-2) ;
                \draw[int] (m-2-3.center) -- (m-3-5) ;
                \draw[int] (m-3-2) -- (m-5-2);
                \draw[int] (m-5-2) -- (m-7-2);
                \draw[int] (m-3-5) -- (m-4-5.center);
                \draw (m-4-5.center) -- (m-5-4) ;
                \draw[int] (m-4-5.center) -- (m-5-6) ;
                \draw[int] (m-5-6) -- (m-6-6.center);
                \draw (m-6-6.center) -- (m-7-5) ;
                \draw (m-6-6.center) -- (m-7-7) ;
                \draw[int] (m-7-2) -- (m-8-2.center);
                \draw (m-8-2.center) -- (m-9-1) ;
                \draw (m-8-2.center) -- (m-9-3) ;
                \draw (m-9-1) -- (m-11-1);
                \draw (m-9-3) -- (m-11-3);
                \draw (m-5-4) -- (m-7-4);
                \draw (m-7-4) -- (m-9-4);
                \draw (m-9-4) -- (m-11-4);
                \draw (m-7-5) -- (m-9-5);
                \draw (m-9-5) -- (m-11-5);
                \draw (m-7-7) -- (m-9-7);
                \draw (m-9-7) -- (m-11-7);
            \end{tikzpicture}
            ~~~~~~ 
            \begin{tikzpicture}[baseline={([yshift=-.5ex]current bounding box.center)}]
                \matrix (m) [matrix of nodes, ampersand replacement=\&, column sep = 0.03cm, row sep = 0.2cm]{
                    {} \& {} \& {} \& {} \& {} \& {} \& {} \& {} \\
                    {} \& {} \& {} \& {} \& {} \& {} \& {} \& {} \\
                    {} \& {} \& {} \& {} \& {} \& {} \& {} \& {} \\
                    {} \& {} \& {} \& {} \& {} \& {} \& {} \& {} \\
                    {} \& {} \& {} \& {} \& {} \& {} \& {} \& {} \\
                    {} \& {} \& {} \& {} \& {} \& {} \& {} \& {} \\
                    {} \& {} \& {} \& {} \& {} \& {} \& {} \& {} \\
                    {} \& {} \& {} \& {} \& {} \& {} \& {} \& {} \\
                    {} \& {} \& {} \& {} \& {} \& {} \& {} \& {} \\
                    {} \& {} \& {} \& {} \& {} \& {} \& {} \& {} \\
                    $\phi_1$ \& $\phi_2$ \& {} \& $\phi_3$ \& $\phi_4$ \& {} \& $\phi_5$ \\
                };
                \draw[anti] (m-1-3) -- (m-2-3.center) ;
                \draw[int] (m-2-3.center) -- (m-3-2) ;
                \draw[int] (m-2-3.center) -- (m-3-6) ;
                \draw[int] (m-3-6) -- (m-5-6);
                \draw[int] (m-5-6) -- (m-7-6);
                \draw[int] (m-7-6) -- (m-8-6.center);
                \draw (m-8-6.center) -- (m-9-5) ;
                \draw (m-8-6.center) -- (m-9-7) ;
                \draw[int] (m-3-2) -- (m-4-2.center);
                \draw (m-4-2.center) -- (m-5-1) ;
                \draw[int] (m-4-2.center) -- (m-5-3) ;
                \draw[int] (m-5-3) -- (m-6-3.center);
                \draw (m-6-3.center) -- (m-7-2) ;
                \draw (m-6-3.center) -- (m-7-4) ;
                \draw (m-5-1) -- (m-7-1);
                \draw (m-7-1) -- (m-9-1);
                \draw (m-9-1) -- (m-11-1);
                \draw (m-7-2) -- (m-9-2);
                \draw (m-9-2) -- (m-11-2);
                \draw (m-7-4) -- (m-9-4);
                \draw (m-9-4) -- (m-11-4);
                \draw (m-9-5) -- (m-11-5);
                \draw (m-9-7) -- (m-11-7);
            \end{tikzpicture}
        \end{equation}
        Thus, once the shape of the tree level HPL diagram and its input legs are fixed, it is produced in precisely one way from the recursion formula~\eqref{eq:tree_recursion} of the homological perturbation lemma.
        
        Each $(n+1)$-ary vertex produced by $\Iint$ comes with a symmetry factor of $\frac{1}{n!}$: $\Iint$ contains $\sfD_\rmint$, which is the sum of the extended, grade-shifted $\sfm_n$ and the latter come with this factor, cf.~\eqref{eq:rel_mu_m}. This factor is compensated by the fact that the input legs of any vertex can be permuted in an arbitrary way. This may require upwards and downwards shifts of subtrees in order to avoid the cancellation by~\eqref{eq:vanish_diag} as exemplified in the first two diagrams of~\eqref{eq:tree_examples}. Also, it may lead to a permutation of the input legs of the diagram; but any permutation of the input legs of the diagram contributes equally with a factor of~1, cf.~equation~\eqref{eq:scattering_amplitude_tree_level}.
    \end{proof}
    
    Having trees under control, we can regard general Feynman diagrams as trees involving one particle irreducible Feynman diagrams\footnote{i.e.~Feynman diagrams which do not become disconnected when removing a single arbitrary edge} as nodes.
    \begin{lemma}
        If the HPL reproduces all one particle irreducible (1PI) Feynman diagrams correctly, it reproduces all Feynman diagrams correctly.
    \end{lemma}
    \begin{proof}
        The proof is evident from the proof of lemma~\ref{lem:1}: any permutation of input legs of vertices connected to subtrees leads to a diagram which is produced in a unique way from the recursion relation~\eqref{eq:quantum_recursion} up to arbitrariness stemming from the 1PI diagrams.
    \end{proof}
    
    Next, we can remove all incoming external legs without affecting the combinatorics:
    \begin{lemma}
        If the HPL reproduces all 1PI tadpole diagrams\footnote{i.e.~diagrams with a single external leg} correctly, then it produces all 1PI diagrams correctly.
    \end{lemma}
    \begin{proof}
        The proof here is simply the realization that in step~2 of algorithm~\ref{alg:main}, the $k$ external legs connected to a vertex with $r$ incoming legs can be attached in any of $\binom{r}{k}$ ways without the diagram becoming vanishing. Thus, attaching $k$ external legs reduces the combinatorial factor of $\frac{1}{r!}$ of the vertex to $\frac{1}{k!}\frac{1}{(r-k)!}$. The first factor is compensated by the fact that any permutation of the external legs of the diagram contributes with a factor of~1 by~\eqref{eq:scattering_amplitude_tree_level}. The remaining second factor would have been the factor of the vertex if the external legs had not been there to begin with.
    \end{proof}   
    
    The first advantage of considering exclusively tadpole diagrams is that we do not have to consider permutations of input legs. The second advantage is that we no longer distinguish between lines corresponding to interacting and free fields: all lines encode interacting fields, and any right outgoing leg produced by $\sfU$ and any outgoing leg produced by $\Iint$ to the left of an existing line makes the diagram vanish by~\eqref{eq:vanish_diag}. Both the horizontal (i.e.~contracting the input legs of which vertices) and vertical (i.e.~order in the recursion relation~\eqref{eq:quantum_recursion}) position of the action of $\sfU$ is then unique for 1PI tadpole diagrams.
    \begin{lemma}
        Given a non-vanishing HPL diagram, corresponding to a 1PI tadpole diagram, with all loop closures removed, any permutation of the order in which the loops are closed and any change in the sequence of the order in which loops and vertices are created leads to a vanishing HPL diagram by~\eqref{eq:vanish_diag}.
    \end{lemma}
    \begin{proof}
        Pictorially, this is immediately clear:
        \begin{equation}
            \begin{tikzpicture}[baseline={([yshift=0ex]current bounding box.center)}]
                \matrix (m) [matrix of nodes, ampersand replacement=\&, column sep = 0.1cm, row sep = 0.2cm]{
                    {} \& {} \& {} \& {} \\
                    {} \& {} \& {} \& {} \\
                    {} \& {} \& {} \& {} \\
                };
                \draw[int] (m-1-4) -- (m-2-4.center) ;
                \draw[int] (m-2-4.center) -- (m-2-1) ;
                \draw[int] (m-1-3) -- (m-3-3.center) ;
                \draw[int] (m-3-3.center) -- (m-3-1);                
            \end{tikzpicture}
            ~~ \rightarrow ~~
            \begin{tikzpicture}[baseline={([yshift=0ex]current bounding box.center)}]
                \matrix (m) [matrix of nodes, ampersand replacement=\&, column sep = 0.1cm, row sep = 0.2cm]{
                    {} \& {} \& {} \& {} \\
                    {} \& {} \& {} \& {} \\
                    {} \& {} \& {} \& {} \\
                };
                \draw[int] (m-1-4) -- (m-3-4.center) ;
                \draw[int] (m-3-4.center) -- (m-3-1) ;
                \draw[int] (m-1-3) -- (m-2-3.center) ;
                \draw[int] (m-2-3.center) -- (m-2-1);                
            \end{tikzpicture}~~=0
            ~~~~~~\mbox{and}~~~~~~
            \begin{tikzpicture}[baseline={([yshift=0ex]current bounding box.center)}]
                \matrix (m) [matrix of nodes, ampersand replacement=\&, column sep = 0.1cm, row sep = 0.2cm]{
                    {} \& {} \& {} \& {} \\
                    {} \& {} \& {} \& {} \\
                    {} \& {} \& {} \& {} \\
                    {} \& {} \& {} \& {} \\
                };
                \draw[int] (m-1-2) -- (m-3-2.center) ;
                \draw[int] (m-3-2.center) -- (m-4-1) ;
                \draw[int] (m-3-2.center) -- (m-4-3) ;
                \draw[int] (m-1-4) -- (m-2-4.center) ;
                \draw[int] (m-2-4.center) -- (m-2-1);                
            \end{tikzpicture}
            ~~ \rightarrow ~~
            \begin{tikzpicture}[baseline={([yshift=0ex]current bounding box.center)}]
                \matrix (m) [matrix of nodes, ampersand replacement=\&, column sep = 0.1cm, row sep = 0.2cm]{
                    {} \& {} \& {} \& {} \\
                    {} \& {} \& {} \& {} \\
                    {} \& {} \& {} \& {} \\
                    {} \& {} \& {} \& {} \\
                };
                \draw[int] (m-1-2) -- (m-2-2.center) ;
                \draw[int] (m-2-2.center) -- (m-3-1.center) ;
                \draw[int] (m-2-2.center) -- (m-3-3.center) ;
                \draw[int] (m-3-1.center) -- (m-4-1.south) ;
                \draw[int] (m-3-3.center) -- (m-4-3.south) ;
                \draw[int] (m-1-4) -- (m-4-4.center) ;
                \draw[int] (m-4-4.center) -- (m-4-1.west);                
            \end{tikzpicture}~~=0~.
        \end{equation}
    \end{proof}
    
    Next, we consider loops more closely, in particular loops involving a single vertex of type~\eqref{diag:self_loops}.
    \begin{lemma}
        If the HPL reproduces all 1PI tadpole diagrams without loops involving a single vertex, it reproduces all 1PI tadpole diagrams.
    \end{lemma}
    \begin{proof}
        These loops are produced in step~3 of the algorithm, and we may have a partially drawn HPL diagrams before and after closing the loops:
        \begin{equation}
            \begin{tikzpicture}[baseline={([yshift=-.5ex]current bounding box.center)}]
                \matrix (m) [matrix of nodes, ampersand replacement=\&, column sep = 0.1cm, row sep = 0.2cm]{
                    {} \& $\phi_{n+1}$ \& {} \& {} \& {} \& {} \& {} \\
                    {} \& {} \& {} \& {} \& {} \& {} \& {} \\
                    {} \& {} \& $\ddots$ \& {} \& {} \& {} \& \\
                    {} \& {} \& {} \& {} \& {} \& {} \& {} \\
                    {} \& {} \& {} \& {} \& {} \& {} \& {} \\
                    {} \& {} \& {} \& {} \& {} \& {} \& {} \\
                    {} \& {} \& {} \& {} \& {} \& {} \& {} \\
                    {} \& {} \& {} \& {} \& {} \& {} \& {} \\
                };
                \draw[anti] (m-1-2) -- (m-2-2.center) ;
                \draw[int] (m-2-2.center) -- (m-3-1) ;
                \draw[int] (m-2-2.center) -- (m-3-3) ;
                \draw[int] (m-4-4) -- (m-5-4.center);
                \draw[int] (m-5-4.center) -- (m-6-2.center) ;
                \draw[int] (m-5-4.center) -- (m-6-3.center) ;
                \draw[int] (m-5-4.center) -- (m-6-4.center) ;
                \draw[int] (m-5-4.center) -- (m-6-5.center) ;
                \draw[int] (m-5-4.center) -- (m-6-6.center) ;
            \end{tikzpicture}
            ~~~\rightarrow~~~
            \begin{tikzpicture}[baseline={([yshift=-.5ex]current bounding box.center)}]
                \matrix (m) [matrix of nodes, ampersand replacement=\&, column sep = 0.1cm, row sep = 0.2cm]{
                    {} \& $\phi_{n+1}$ \& {} \& {} \& {} \& {} \& {} \\
                    {} \& {} \& {} \& {} \& {} \& {} \& {} \\
                    {} \& {} \& $\ddots$ \& {} \& {} \& {} \& \\
                    {} \& {} \& {} \& {} \& {} \& {} \& {} \\
                    {} \& {} \& {} \& {} \& {} \& {} \& {} \\
                    {} \& {} \& {} \& {} \& {} \& {} \& {} \\
                    {} \& {} \& {} \& {} \& {} \& {} \& {} \\
                    {} \& {} \& {} \& {} \& {} \& {} \& {} \\
                };
                \draw[anti] (m-1-2) -- (m-2-2.center) ;
                \draw[int] (m-2-2.center) -- (m-3-1) ;
                \draw[int] (m-2-2.center) -- (m-3-3) ;
                \draw[int] (m-4-4) -- (m-5-4.center);
                \draw[int] (m-5-4.center) -- (m-6-2.center) ;
                \draw[int] (m-5-4.center) -- (m-6-3.center) ;
                \draw[int] (m-5-4.center) -- (m-6-4.center) ;
                \draw[int] (m-5-4.center) -- (m-6-5.center) ;
                \draw[int] (m-5-4.center) -- (m-6-6.center) ;
                \draw[int] (m-6-3) -- (m-8-3.center);
                \draw[int] (m-6-5) -- (m-8-5.center);
                \draw[int] (m-8-3.center) -- (m-8-5.center);
                \draw[int] (m-6-4) -- (m-7-4.center);
                \draw[int] (m-6-6) -- (m-7-6.center);
                \draw[int] (m-7-4.center) -- (m-7-6.center);
            \end{tikzpicture}
        \end{equation}
        Say that the vertex involving the loops has $r$ open input legs and we have $\ell$ loops to close off. Then we have $\binom{r}{2}$ choices for the first loop, $\binom{r-2}{2}$ choices for the second loop, etc. Also, pairing off the loops with the operator $\sfU$ only leads to a non-vanishing HPL diagram, if the right leg generated by $\sfU$ is to the right of all previously generated right legs. That is, there is a unique order in which the loops can be closed, and we have to divide our counting by an additional factor of $\frac{1}{\ell!}$. Altogether, we have
        \begin{equation}
            \frac{1}{r!}\binom{r}{2}\binom{r-2}{2}\ldots \binom{r+2-2\ell}{2}\frac{1}{\ell!}=\frac{1}{\ell! 2^\ell}\times\frac{1}{(r-2\ell)!}~.
        \end{equation}
        The first factor is precisely the expected factor in~\eqref{eq:formula_Sigma} and the second factor would have been the factor of the vertex if the loops we closed had not been there to begin with.
    \end{proof}
    
    Finally, we can conclude with the following lemma: 
    \begin{lemma} 
        All 1PI tadpole diagrams without loops involving a single vertex are reproduced with the right symmetry factor.
    \end{lemma}
    \begin{proof}
        This is now readily seen from a comparison of the HPL diagrams with algorithm~\ref{alg:main} for constructing diagrams from the generating functional. We note that in the generating functional approach in steps 3 and 4 we can use arbitrary input legs for closing loops or attaching vertices, while in the HPL diagrams, there is an order: all legs connected to input legs of vertices further below in the diagram have to be to the left of the legs leading to output legs of other vertices or to input legs of vertices further up in the diagram, for example:
        \begin{equation}
            \begin{tikzpicture}[baseline={([yshift=0ex]current bounding box.center)}]
                \matrix (m) [matrix of nodes, ampersand replacement=\&, column sep = 0.3cm, row sep = 0.2cm]{
                    {} \& {} \& {} \& {} \& {} \\
                    {} \& {} \& {} \& {} \& {} \\
                    {} \& {} \& {} \& {} \& {} \\
                    {} \& {} \& {} \& {} \& {} \\
                    {} \& {} \& {} \& {} \& {} \\
                    {} \& {} \& {} \& {} \& {} \\
                    {} \& {} \& {} \& {} \& {} \\
                    {} \& {} \& {} \& {} \& {} \\
                };
                \draw[int] (m-1-3) -- (m-2-3.center) ;
                \draw[int] (m-2-3.center) -- (m-3-1.center);                
                \draw[int] (m-2-3.center) -- (m-3-2.center);                
                \draw[int] (m-2-3.center) -- (m-3-3.center);                
                \draw[int] (m-2-3.center) -- (m-3-4.center);                
                \draw[int] (m-2-3.center) -- (m-3-5.center);                
                \draw[int] (m-3-1) -- (m-8-1.center);                
                \draw[int] (m-3-2) -- (m-8-2.center);                
                \draw[int] (m-3-4) -- (m-6-4);                
                \draw[int] (m-6-4) -- (m-7-4.center) ;
                \draw[int] (m-7-4.center) -- (m-8-3.center);                
                \draw[int] (m-7-4.center) -- (m-8-5.center);                
               \draw[int] (m-3-3) -- (m-5-3.center);                
               \draw[int] (m-5-3.center) -- (m-5-1.west);                
               \draw[int] (m-3-5) -- (m-4-5.center);                
               \draw[int] (m-4-5.center) -- (m-4-1.west);                
            \end{tikzpicture}
        \end{equation}
        For a vertex with $n$ input legs, $k$ of which link to outgoing legs of vertices below, the total combinatorial factor for the allowed permutations is thus $\frac{n!}{k!(n-k)!}=\binom{n}{k}$, and it is this non-trivial factor that is responsible for producing the remaining symmetry factors. 
        
        The ordering of the types of input legs for each vertex by the HPL now simply ensures that identical contractions in the computation of correlators by the generating functional  approach~\eqref{eq:correlators} are not counted twice. For example, the following two HPL-diagrams describe the same contraction of functional derivatives and propagators in~\eqref{eq:correlators}:
        \begin{equation}
            \begin{tikzpicture}[baseline={([yshift=-.5ex]current bounding box.center)}]
                \matrix (m) [matrix of nodes, ampersand replacement=\&, column sep = 0.03cm, row sep = 0.2cm]{
                    {} \& $\phi_1$ \& {} \& {} \\
                    {} \& {} \& {} \& {} \\
                    {} \& {} \& {} \& {} \\
                    {} \& {} \& {} \& {} \\
                    {} \& {} \& {} \& {} \\
                    {} \& {} \& {} \& {} \\
                };
                \draw[int] (m-1-2) -- (m-2-2.center) ;
                \draw[int] (m-2-2.center) -- (m-3-1) ;
                \draw[int] (m-2-2.center) -- (m-3-3) ;
                \draw[int] (m-3-1) -- (m-5-1);
                \draw[int] (m-3-3) -- (m-4-3.center);
                \draw[int] (m-4-3.center) -- (m-5-2) ;
                \draw[int] (m-4-3.center) -- (m-5-4) ;
                \draw[int] (m-5-2) -- (m-6-2.center) ;
                \draw[int] (m-5-1) -- (m-6-1.center) ;
                \draw[int] (m-6-1.center) -- (m-6-2.center) ;
                \draw[int] (m-6-4.south) -- (m-5-4) ;
            \end{tikzpicture}
            ~~~\mbox{and}~~~
            \begin{tikzpicture}[baseline={([yshift=-.8ex]current bounding box.center)}]
                \matrix (m) [matrix of nodes, ampersand replacement=\&, column sep = 0.03cm, row sep = 0.2cm]{
                    {} \& {} \& $\phi_1$ \& {} \\
                    {} \& {} \& {} \& {} \\
                    {} \& {} \& {} \& {} \\
                    {} \& {} \& {} \& {} \\
                    {} \& {} \& {} \& {} \\
                    {} \& {} \& {} \& {} \\
                };
                \draw[int] (m-1-3) -- (m-2-3.center) ;
                \draw[int] (m-2-3.center) -- (m-3-2) ;
                \draw[int] (m-2-3.center) -- (m-3-4) ;
                \draw[int] (m-3-4) -- (m-5-4);
                \draw[int] (m-3-2) -- (m-4-2.center);
                \draw[int] (m-4-2.center) -- (m-5-1) ;
                \draw[int] (m-4-2.center) -- (m-5-3) ;
                \draw[int] (m-5-3) -- (m-6-3.center) ;
                \draw[int] (m-5-4) -- (m-6-4.center) ;
                \draw[int] (m-6-3.center) -- (m-6-4.center) ;
                \draw[int] (m-6-1.south) -- (m-5-1) ;
            \end{tikzpicture}
        \end{equation}
        but should be counted only once; the HPL does this by letting the right HPL-diagram vanish by~\eqref{eq:vanish_diag}. In this manner, uniqueness of the closure of loops (corresponding to the same contraction) is guaranteed by the asymmetry of $\sfH_0$, which enforces the closure of all loops to the bottom left of all involved vertices in the HPL diagrams.
        
        Altogether, the HPL reproduces all Feynman diagrams corresponding to one particular pairing of functional derivatives with sources in~\eqref{eq:vanish_diag} precisely once.
    \end{proof}
    
    \section*{Acknowledgements}
    
    We are very grateful to Martin Wolf for comments on a first draft of this paper. The work of ES was supported by an Undergraduate Research Bursary of the London Mathematical Society.

    \bibliography{bigone}
    
    \bibliographystyle{latexeu}
    
\end{document}